\documentclass[a4paper,UKenglish]{lipics-v2018}

\usepackage{microtype}

\title{Complexity of Reconfiguration Problems for Constraint Satisfaction}

\author{Tatsuhiko Hatanaka}{Tohoku University, Japan}{hatanaka@ecei.tohoku.ac.jp}{}{}
\author{Takehiro Ito}{Tohoku University, Japan}{takehiro@ecei.tohoku.ac.jp}{}{}
\author{Xiao Zhou}{Tohoku University, Japan}{zhou@ecei.tohoku.ac.jp}{}{}

\authorrunning{T. Hatanaka, T. Ito and X. Zhou} 

\Copyright{Tatsuhiko Hatanaka, Takehiro Ito, and Xiao Zhou}

\subjclass{Mathematics of computing $\rightarrow$ Graph algorithms }


\category{}

\relatedversion{}

\supplement{}

\funding{}

\acknowledgements{This work is partially supported by JST CREST Grant Number JPMJCR1402, and by JSPS KAKENHI Grant Numbers JP18H04091, JP16J02175, JP16K00003, and JP16K00004, Japan.}

\hideLIPIcs  

\usepackage{comment}

\newcommand{\Neigh}[2]{N(#1,#2)}
\newcommand{\NN}{\mathbb{N}}

\newcounter{Ctemp}

\newcommand{\ISN}[2]{{
	\setcounter{Ctemp}{#1}
	\stepcounter{Ctemp}
	\{#1,\theCtemp,\ldots,#2\}}}

\newcommand{\VC}{\mathsf{vc}}
\newcommand{\TD}{\mathsf{td}}
\newcommand{\MW}{\mathsf{mw}}
\newcommand{\PW}{\mathsf{pw}}
\newcommand{\BW}{\mathsf{bw}}
\newcommand{\TW}{\mathsf{tw}}
\newcommand{\CW}{\mathsf{cw}}
\newcommand{\NB}{\mathsf{nb}}

\newcommand{\hpE}{E}
\newcommand{\Prm}[1]{\mathcal{P}(#1)}

\newcommand{\Vcns}{\mathcal{L}}
\newcommand{\Ecns}{\mathcal{C}}
\newcommand{\Dom}{D}
\newcommand{\doms}{k}
\newcommand{\ari}{r}
\newcommand{\underlying}{underlying}

\newcommand{\Sol}[1]{\mathscr{S}(#1)}

\newcommand{\mapf}{f}
\newcommand{\mapg}{g}
\newcommand{\maph}{h}
\newcommand{\ini}{{s}}
\newcommand{\tar}{{t}}
\newcommand{\iot}{{r}}
\newcommand{\Seq}{\mathcal{W}}

\newcommand{\CSP}[1]{\ifx #10\textsc{constraint satisfiability}\else $#1$\textsc{-constraint satisfiability}\fi}

\newcommand{\CSR}[1]{\ifx #10\textsc{CSR}\else $#1$\textsc{-CSR}\fi}
\newcommand{\BCSR}{\textsc{BCSR}}
\newcommand{\CR}[1]{\textsc{#1CR}}
\newcommand{\HR}[1]{\textsc{#1HR}}

\newcommand{\Rels}{\mathcal{S}}

\newcommand{\rest}[2]{#1|_{#2}}
\newcommand{\diff}[2]{\mathsf{dif}(#1,#2)}
\newcommand{\Inst}{\mathcal{I}}
\newcommand{\seq}[1]{\langle #1\rangle}

\newcommand{\Cls}[2]{\mathscr{C}_{#1}(#2)}
\newcommand{\tk}{{\tau}}

\newcommand{\kk}{\kappa}
\newcommand{\kkc}{$(\kk \times \kk)$\textsc{-clique}}
\newcommand{\kkv}[2]{u_{#1}^{#2}}
\newcommand{\kkgp}[1]{U_{#1}}

\newcommand{\runtime}[2]{O^*(#1^{o(#2)})}

\newcommand{\Hwg}{R}
\newcommand{\Hws}[2]{\mathscr{W}_{#1}(#2)}
\newcommand{\wdl}{\rho}
\newcommand{\Hw}{\boldsymbol{w}}
\newcommand{\lay}[1]{L^{#1}}
\newcommand{\Elay}[1]{E^{#1}}

\newcommand{\asgn}[1]{\ifx #10\mathcal{A}\else \mathcal{A}(#1)\fi}
\newcommand{\isov}{\phi}
\newcommand{\isoe}{\pi}
\newcommand{\repcns}[2]{#1[#2]}

\newcommand{\subtree}[1]{T_{#1}}
\newcommand{\ances}[1]{\mathsf{Anc}(#1)}

\newcommand{\IDN}[1]{\mathscr{N}(#1)}
\newcommand{\IDA}[1]{\mathscr{A}(#1)}
\newcommand{\IDC}[1]{\mathscr{C}(#1)}
\newcommand{\tdkern}[1]{h_{\doms,\TD}(#1)}

\newcommand{\Part}{\mathscr{P}}
\newcommand{\CSG}[2]{\mathsf{CSG}(#1,#2)}
\newcommand{\Cnd}{P}
\newcommand{\eqcls}[1]{[#1]}
\newcommand{\Instcs}{\mathcal{J}}
\newcommand{\sbst}[2]{\mathsf{SUB}(#1;#2)}
\newcommand{\subcns}[1]{\mathcal{G}(#1)}

\newcommand{\weight}{\omega}
\newcommand{\Solw}[2]{\mathscr{S}_{#1}(#2)}
\newcommand{\len}[2]{\mathsf{len}_{#1}(#2)}
\newcommand{\nrec}[2]{\# (#1,#2)}
\newcommand{\OPT}[2]{\mathsf{OPT}(#1,#2)}

\newcommand{\Cov}{C}
\newcommand{\Vind}{I}

\newcommand{\Bs}{V_\mathsf{B}}
\newcommand{\NBs}{V_\mathsf{N}}
\newcommand{\Imp}[1]{\mathsf{IMP}(#1)}
\newcommand{\vimp}[2]{#1[#2]}
\newcommand{\resup}{\mathsf{res}}
\newcommand{\ppsup}{{*}}
\newcommand{\Vfix}{V_\mathsf{fix}}

\newtheorem{proposition}{Proposition}
\newtheorem{observation}{Observation}
\newtheorem{claim}{Claim}
\nolinenumbers

\begin{document}

\maketitle

\begin{abstract}
	Constraint satisfaction problem (CSP) is a well-studied combinatorial search problem, in which we are asked to find an assignment of values to given variables so as to satisfy all of given constraints.
	We study a reconfiguration variant of CSP, in which we are given an instance of CSP together with its two satisfying assignments, and asked to determine whether one assignment can be transformed into the other by changing a single variable assignment at a time, while always remaining satisfying assignment.
	This problem generalizes several well-studied reconfiguration problems such as Boolean satisfiability reconfiguration, vertex coloring reconfiguration, homomorphism reconfiguration.
	In this paper, we study the problem from the viewpoints of polynomial-time solvability and parameterized complexity, and give several interesting boundaries of tractable and intractable cases. 
 \end{abstract}

\newpage
\section{Introduction}

	Recently, the framework of \emph{reconfiguration}~\cite{IDHPSUU} has been extensively studied in the field of theoretical computer science.
	This framework models several ``dynamic'' situations where we wish to find a step-by-step transformation between two feasible solutions of a combinatorial (search) problem such that all intermediate solutions are also feasible and each step respects a fixed reconfiguration rule.
	This reconfiguration framework has been applied to several well-studied combinatorial problems.
	(See surveys~\cite{ISweb,Nis18,Jan13}.)

\begin{figure}[t]
	\begin{center}
		\includegraphics{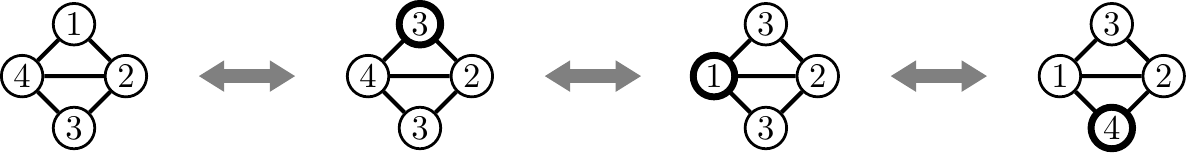}
	\end{center}
	\caption{A transformation of $4$-colorings.
	A vertex which is recolored from the previous $4$-coloring is depicted by a thick circle.}
	\label{fig:coloring}
\end{figure}

	As an example, we consider the (vertex) \textsc{coloring reconfiguration} problem~\cite{BC09,BP16,CHJ11,HIZ17c,Wro18}, which is one of the most well-studied reconfiguration problems.
	Let $G=(V,E)$ be a graph and let $\Dom$ be a set of $\doms$ \emph{colors}.
	A $\doms$-\emph{coloring} of $G$ is a mapping $\mapf \colon V \to \Dom$ such that $\mapf(v)\ne \mapf(w)$ holds for every edge $vw\in E$.
	In the reconfiguration framework, we wish to transform one $\doms$-coloring into another one by recoloring a single vertex at a time, while always remaining $\doms$-coloring.
	(See \figurename~\ref{fig:coloring} for example.)

\begin{figure}[t]
	\begin{center}
		\includegraphics{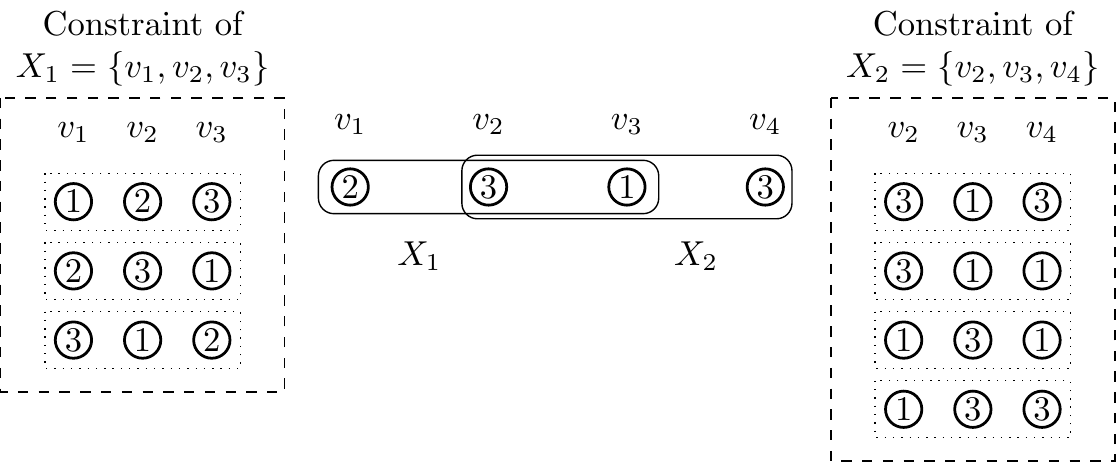}
	\end{center}
	\caption{An example of constraints which represent allowed assignments to the vertices in CSP (left and right of the figure) and a mapping which satisfies all constraints (middle of the figure).}
	\label{fig:CSP}
\end{figure}
	
	In this paper, we study a reconfiguration variant of the well-known constraint satisfaction problem (CSP, for short), which is a generalization of several combinatorial problems including vertex coloring, Boolean satisfiability, graph homomorphism.
	The formal definitions of CSP and its reconfiguration variant will be given in Section~\ref{sec:problem}, but we here briefly introduce them.
	Let $G=(V,\hpE)$ be a hypergraph.
	Let $\Dom$ be a set, called a \emph{domain}; each element of $\Dom$ is called a \emph{value} and we always denote by $\doms$ the size of a domain.
	In CSP, each hyperedge $X\in \hpE$ has a \emph{constraint} which represents the values allowed to be assigned to the vertices in $X$ at the same time, and we wish to find a mapping $V \to \Dom$ which satisfies the constraints of all hyperedges in $G$. 
	(See \figurename~\ref{fig:CSP} for an example.)
	In the case of vertex coloring, we can see that every hyperedge consists of two vertices, and has the common constraint that any two different colors can be assigned to the two vertices in the hyperedge at the same time.  
	In the reconfiguration framework for CSP, we are given a CSP instance together with its two mappings satisfying all constraints, and we wish to transform one mapping into the other by changing a value of a single vertex at a time, while always satisfying all constraints.
	
	The reconfiguration problem for CSP has been studied as several spacial cases such as 
		\textsc{Boolean constraint satisfiability reconfiguration}~\cite{BMNR14,CDEHW18,GKMP09,MTY10,MTY11,MNPR17,Schwerdtfeger14}, 
		\textsc{homomorphism reconfiguration}~\cite{BLS18,BLMNS17,BMMN16,BN15,Wro15,Wro18}, 
		and \textsc{(list) coloring reconfiguration}~\cite{BB18,BJLPP14,BC09,BMNR14,BP16,BousquetP16,CCMS16,Cer07,CHJ11,DFFV06,FJP16,HM18,HIZ14,HIZ17c,HIZ18,IKZ12,JKKPP14,OSIZ17,OSIZ18,Wro18}.
	We will briefly summarize these results in the next section.
	
	In this paper, we study the reconfiguration problem for CSP and its special cases from the viewpoints of polynomial-time solvability and parameterized complexity, and give several interesting boundaries of tractable and intractable cases.
	Our results are roughly summarized as follows:
	We first investigate the complexity of the problem with respect to the size $\doms$ of a domain (summarized in \tablename~\ref{tab:result_k}).
	We next analyze the (parameterized) complexities with respect to the structure of an input hypergraph (summarized in \tablename s~\ref{tab:result_pw} and \ref{tab:result_graph}).
	We finally explore the boundary of time complexity of the problem; we give some exact algorithms, and a lower bound of the computation time that meets the running times of our algorithms under the exponential time hypothesis.

\section{Problem definition and results}
\label{sec:problem}

\subsection{Our problem}

	In this subsection, we formally define CSP and its reconfiguration variant by means of hypergraphs.
	
	We first give some terminologies regarding hypergraphs and mappings.	
	A \emph{hypergraph} $G$ is a pair $(V,\hpE)$, where $V$ is a set of vertices and $\hpE$ is a family of non-empty vertex subsets, called \emph{hyperedges}.
	A hypergraph $G$ is $r$\emph{-uniform} if every hyperedge consists of exactly $r$ $(\ge 1)$ vertices.	
	Sometimes, a $2$-uniform hypergraph $G$ is simply called a \emph{graph} and each hyperedge of $G$ is called an \emph{edge}.
	An edge $\{v,w\}$ is sometimes denoted by $vw$ or $wv$ for notational convenience.
	Let $A$ and $B$ be any sets.
	We denote by $B^A$ the set of all mappings from $A$ to $B$, because we can identify a mapping $\phi \colon A \to B$ with a vector $(\phi(a_1), \phi(a_2), \ldots, \phi(a_{|A|})) \in B^{|A|}$, where $A=\{a_1,a_2,\ldots,a_{|A|}\}$.
	Let $\phi \in B^A$ be a mapping.
	For any subset $A^\prime$ of $A$, we denote by $\rest{\phi}{A^\prime}$ the restriction of $\phi$ on $A^\prime$; that is, $\rest{\phi}{A^\prime}$ is a mapping in $B^{A^\prime}$ such that $\rest{\phi}{A^\prime}(a)=\phi(a)$ for any $a\in A^\prime$.
	
	\medskip
	We next define CSP and give some notation.
	Let $G=(V,\hpE)$ be a hypergraph.
	Let $\Dom$ be a set, called a \emph{domain}; each element of $\Dom$ is called a \emph{value} and we always denote by $\doms$ the size (cardinality) of a domain.
	In CSP, each hyperedge $X\in \hpE$ has a set $\Ecns(X) \subseteq \Dom^X$; we call $\Ecns(X)$ a \emph{constraint of} $X$.
	If $\Ecns(X)=\Dom^X$, it is called a \emph{trivial} constraint.
	The \emph{arity} of a constraint $\Ecns(X)$ of $X$ is exactly $|X|$, and we call $\Ecns(X)$ an $\ari$\emph{-ary} constraint, where $\ari=|X|$.
	We define the \emph{constraint} $\Ecns(G)$ \emph{of} $G$ as the union of all constraints of hyperedges, that is, $\Ecns(G)=\bigcup_{X\in E(G)} \Ecns(X)$.
	For a vertex $v\in V(G)$, a \emph{list} $\Vcns(v)$ of $v$ is the set $\{i\in \Dom \colon \exists \mapg \in \Ecns(G),\ \mapg(v)=i\}$.
	For a hyperedge $X\in \hpE$, we say that a mapping $\mapf \in \Dom^V$ \emph{satisfies} a constraint of $X$ if $\rest{\mapf}{X}\in \Ecns(X)$ holds.
	$\mapf$ is a \emph{solution} if it satisfies all constraints.
	
	An instance of \textsc{constraint satisfiability} is a triple $(G,\Dom,\Ecns)$ consisting of a hypergraph $G$, a domain $\Dom$, and a constraint assignment to hyperedges over $\Dom$.
	Then, the problem asks whether there exists a solution or not.
	\textsc{Constraint satisfiability} includes many combinatorial problems as its special cases as follows.
	\begin{itemize}
		\item \textsc{Boolean constraint satisfiability} is the special case of \textsc{constraint satisfiability} where $|\Dom|=2$.
		\item For an integer $\ari \ge 1$, $\ari$\textsc{-ary constraint satisfiability} is the special case of \textsc{constraint satisfiability} where all constraints are of arity at most $\ari$, that is, all hyperedges have size at most $\ari$.
		\item \textsc{List homomorphism} is the special case of $2$\textsc{-ary constraint satisfiability} where $G$ is a $2$-uniform hypergraph (i.e., a graph) and there exists a simple undirected graph $H$ with a vertex set $\Dom$, called an \emph{\underlying{} graph}\footnote{
			In this paper, we only consider simple undirected \underlying{} graphs, although loops and/or directed edges are often allowed in researches related to graph homomorphisms.}, 
		such that $\Ecns(vw)=E(H)\cap(\Vcns(v) \times \Vcns(w))$ holds for every edge $vw\in E(G)$, where $\Vcns(v)$ and $\Vcns(w)$ are lists of $v$ and $w$, respectively.
		In other words, $\Ecns(vw)$ is the set of all homomorphisms from the edge $vw$ to $H$ which respect the lists of $v$ and $w$.
		\item \textsc{Homomorphism} is the special case of \textsc{list homomorphism} where $\Vcns(v)=\Dom$ holds for every vertex $v\in V(G)$.
		\item \textsc{(List) coloring} is the special case of \textsc{(list) homomorphism} where an \underlying{} graph $H$ is complete, that is, $\Ecns(vw)$ is a set of all injective mappings from $\{v,w\}$ to $\Dom$ (which respect the lists of $v$ and $w$).
	\end{itemize}
	
	\medskip
	We then define a reconfiguration variant of \textsc{constraint satisfiability}, that is, \textsc{constraint satisfiability reconfiguration}.
	Let $\mapf$ and $\mapf^\prime$ be two solutions for an instance $\Inst=(G,\Dom,\Ecns)$ of \CSP{0}.
	We define the \emph{difference} $\diff{\mapf}{\mapf^\prime}$ between $\mapf$ and $\mapf^\prime$ as the set $\{v\in V(G)\colon \mapf(v) \ne \mapf^\prime(v) \}$.
	We now define the \emph{solution graph} $\Sol{\Inst}$ for $\Inst$ as follows.
	$V(\Sol{\Inst})$ is the set of all solutions for $\Inst$, and two solutions $\mapf$ and $\mapf^\prime$ are connected by an edge if and only if $|\diff{\mapf}{\mapf^\prime}|=1$.
	A walk in $\Sol{\Inst}$ is called a \emph{reconfiguration sequence}.
	Two solutions $\mapf$ and $\mapf^\prime$ are \emph{reconfigurable} if and only if there exists a reconfiguration sequence between them.
	
	An instance of \textsc{constraint satisfiability reconfiguration} (\CSR{0} for short) is a 5-tuple $(G,\Dom,\Ecns,\allowbreak \mapf_\ini,\mapf_\tar)$, where $(G,\Dom,\Ecns)$ is an instance of \textsc{constraint satisfiability}, and $\mapf_\ini$ and $\mapf_\tar$ are two solutions to $(G,\Dom,\Ecns)$, called \emph{initial} and \emph{target} solutions.
	Then, the problem asks whether $\mapf_\ini$ and $\mapf_\tar$ are reconfigurable or not.
	Similarly, for each special case \textsc{P} of \textsc{constraint satisfiability}, we define ``\textsc{P reconfiguration}'' as a special case of \CSR{0}.
	We use the following abbreviations:
	\begin{itemize}
		\item \BCSR{} for \textsc{Boolean constraint satisfiability reconfiguration};
		\item \CSR{\ari} for $\ari$\textsc{-ary constraint satisfiability reconfiguration} for each integer $\ari \ge 1$;
		\item \HR{(L)} for \textsc{(list) homomorphism reconfiguration}; and
		\item \CR{(L)} for \textsc{(list) coloring reconfiguration}.
	\end{itemize}
	Relationships between problems are illustrated in \figurename~\ref{fig:relation}(a).

	\begin{figure}[t]
		\begin{center}
			\includegraphics{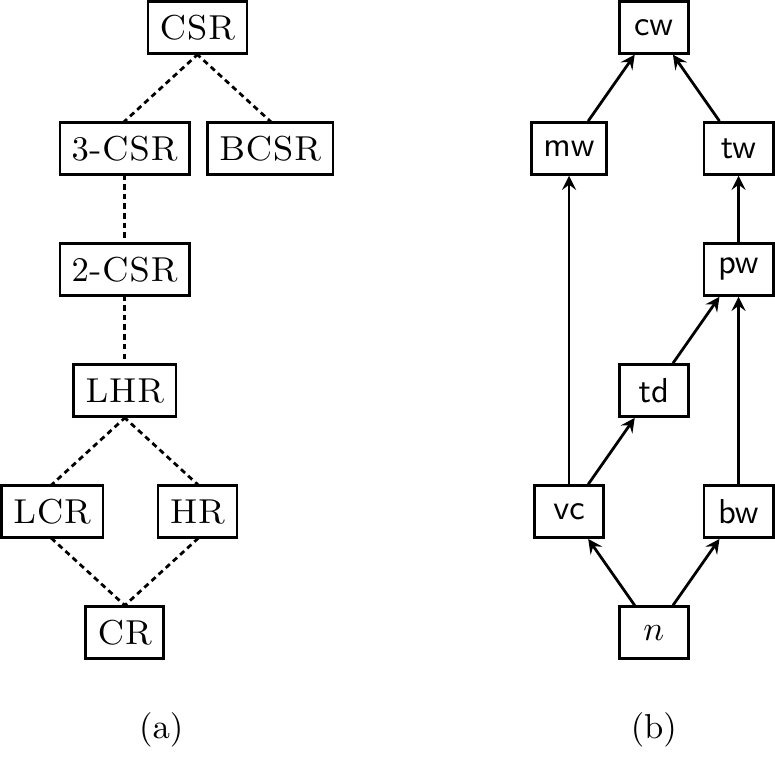}
		\end{center}
		\caption{(a) Relationships between problems.
			Each dotted line between P (lower) and Q (upper) means that P is a special case of Q.
			(b) Relationships between graph parameters.
			$\CW$, $\MW$, $\TW$, $\PW$, $\TD$, $\VC$, $\BW$ and $n$ are the cliquewidth, the modular-width, the treewidth, the pathwidth, the tree-depth, the size of a minimum vertex cover, the bandwidth and the number of vertices of a graph, respectively.
			Each arrow $\alpha \to \beta$ means that $\alpha$ is stronger than $\beta$, that is, if $\alpha$ is bounded by a constant then $\beta$ is also bounded by some constant.}
		\label{fig:relation}
	\end{figure}

\subsection{Known and related results}
	
	There are many literatures which study special cases of \CSR{0} and their \emph{shortest variants}.
	In the shortest variant, we are given an instance with an integer $\ell \ge 0$, and asked whether there exists a reconfiguration sequence of length at most $\ell$.
	We here state only the results from the viewpoint of the computational complexity.
	
	One of the most well-studied special cases of \CSR{0} is \BCSR{}~\cite{BMNR14,CDEHW18,GKMP09,MTY10,MTY11,MNPR17,Schwerdtfeger14}.
	Gopalan et al.~\cite{GKMP09} gave a computational dichotomy for \BCSR{} with respect to a set $\Rels$ of logical relations which can be used to define each constraint;
	the problem is PSPACE-complete or in P for each fixed $\Rels$.
	In addition, Cardinal et al.~\cite{CDEHW18} showed that the problem remains PSPACE-complete even if $\Rels$ is equivalent to monotone Not-All-Equal $3$-SAT (i.e., each constraint is a set of surjections) and a ``variable-clause incidence graph'' is planar.
	For the shortest variant, a computational trichotomy is known; Mouawad et al.~\cite{MNPR17} proved that the problem is PSPACE-complete, NP-complete or in P for each fixed $\Rels$.
	Bonsma et al.~\cite{BMNR14} proved that the shortest variant is $W[2]$-hard when parameterized by $\ell$ even if $\Rels$ is equivalent to Horn SAT.
	
	Another well-studied spacial case is \CR{}~\cite{BB18,BJLPP14,BC09,BMNR14,BP16,BousquetP16,CCMS16,Cer07,CHJ11,DFFV06,FJP16,HM18,HIZ17c,JKKPP14,OSIZ17,OSIZ18,Wro18}.
	This problem is PSPACE-complete for $\doms \ge 4$ and bipartite planar graphs~\cite{BC09} but in P for $\doms \le 3$~\cite{CHJ11}.
	We note that these results (including tractability one) can be extended for \CR{L}.
	Moreover, it is known that the problem remains PSPACE-complete even if $\doms$ is a fixed constant for several graph classes such as line graphs (for any fixed $\doms \ge 5$)~\cite{OSIZ18}, bounded bandwidth graphs~\cite{Wro18}, and chordal graphs~\cite{HIZ17c}.
	On the other hand, several polynomial-time algorithms are known for subclasses of chordal graphs such as trivially perfect graphs, split graphs~\cite{HIZ17c}, and $(\doms-2)$-connected chordal graphs~\cite{BP16}.
	For the shortest variant parameterized by $\ell$, some intractability results are known;
	it is $W[1]$-hard~\cite{BMNR14} and does not admit a polynomial kernelization when $\doms$ is fixed unless the polynomial hierarchy collapses~\cite{JKKPP14}.
	
	As a generalization of \CR{}, \CR{L} is also studied~\cite{HIZ14,HIZ18,IKZ12,OSIZ18,Wro18}.
	The problem is PSPACE-complete even if $\doms$ is a constant for graphs with pathwidth two~\cite{Wro18}, while it is polynomial-time solvable for graphs with pathwidth one~\cite{HIZ14}.
	Hatanaka et al.~\cite{HIZ18} gave fixed-parameter algorithms for \CR{L} parameterized by $\doms+\MW$ and for the shortest variant parameterized by $\doms+\VC$.
	In contrast, they also showed that the problem is $W[1]$-hard when parameterized only by $\VC$.
	
	\HR{} is also well-studied as a generalization of \CR{}.
	Several literatures investigated \HR{} from the viewpoint of graph classes in which a given graph $G$ or an \underlying{} graph $H$ lies~\cite{BLS18,BLMNS17,BMMN16,BN15,Wro15,Wro18}.
	Brewster et al.~\cite{BMMN16} gave a dichotomy for a special case of \HR{} in which $H$ is a $(p,q)$-circular clique;
	it is PSPACE-complete if $p/q\ge 4$ but is in P otherwise.
	We note that this result generalizes the complexity of \CR{} with respect to $\doms$, since a complete graph $K_{\doms}$ is a $(\doms,1)$-circular clique.
	It is also known that the problem is PSPACE-complete even if $H$ is an odd wheel~\cite{BLMNS17} or $H$ is some fixed graph and $G$ is a cycle~\cite{Wro18}.
	On the other hand, it can be solved in polynomial time if $G$ is a tree~\cite{Wro18} or $H$ contains no cycles of length four as a subgraph~\cite{Wro15}.
	Furthermore, a fixed-parameter algorithm is known when parameterized by $\doms+\TD$~\cite{Wro18}; note that it can be easily extended for \HR{L}.
	
	Finally, we refer to the shortest variant of \CSR{0}.
	Bonsma et al.~\cite{BMNR14} gave a fixed-parameter algorithm for the shortest variant parameterized by $\doms+\ari+\ell$, where $\ari$ is the maximum arity of a constraint.
	This implies that shortest variants of \BCSR{} and \CSR{2} are fixed-parameter tractable when parameterized by $\ari+\ell$ and $\doms+\ell$, respectively.
	They also showed that the problem is intractable in general if at least one of $\{\doms,\ari,\ell\}$ is excluded from the parameter.

\subsection{Our contribution}

	In this paper, we investigate the complexity of \CSR{0} and its spacial cases, especially \CSR{3}, \CSR{2}, \HR{(L)} and \CR{(L)}, from several viewpoints.

\paragraph*{The sizes of a domain and lists}

\begin{table}[t]
	\centering
	\caption{Computational complexities with respect to the size $\doms$ of a domain.}
	\medskip
	\label{tab:result_k}
	\begin{tabular}{|l|l|l|l|}
		\hline
		& $\doms \ge 4$ & $\doms=3$ & $\doms=2$ \\ \hline
		\CSR{0}	& PSPACE-c.		& PSPACE-c.		& PSPACE-c.		\\ \hline
		\CSR{3}	& PSPACE-c.		& PSPACE-c.		& PSPACE-c.~\cite{GKMP09}	\\ \hline
		\CSR{2}	& PSPACE-c.		& PSPACE-c.~[Thm.~\ref{the:k3hard}]	& P~[Thm.~\ref{the:k2easy}]		\\ \hline
		\HR{L}	& PSPACE-c.		& P~[Thm.~\ref{the:k3easy}]		& P	\\ \hline
		\CR{L}	& PSPACE-c.		& P~\cite{CHJ11}	& P	\\ \hline
		\HR{}		& PSPACE-c.		& P~\cite{Wro15}	& P	\\ \hline
		\CR{}		& PSPACE-c.~\cite{BC09}	& P	& P	\\ \hline
	\end{tabular}
\end{table}

\begin{table}[t]
	\centering
	\caption{Parameterized complexity with respect to $\doms$ and the number $\NB$ of vertices whose lists have more than two values.}
	\medskip
	\label{tab:result_NB}
	\begin{tabular}{|l|l|l|}
		\hline
		Parameter & $\doms+\NB$  & $\NB$	\\ \hline
		\CSR{0} & PSPACE-c.               & PSPACE-c.	\\ \hline
		\CSR{3} & PSPACE-c.~\cite{GKMP09} & PSPACE-c.	\\ \hline
		\CSR{2} & FPT~[Thm.~\ref{the:nb}] & $W[1]$-hard but XP~[Thm.~\ref{the:nb}]	\\ \hline
		\HR{L}  & FPT                     & $W[1]$-hard but XP	\\ \hline
		\CR{L}  & FPT                     & $W[1]$-hard but XP	\\ \hline
		\HR{}   & FPT                     & $W[1]$-hard~[Cor.~\ref{cor:W1hard}] but XP	\\ \hline
		\CR{}   & FPT                     & FPT	\\ \hline
	\end{tabular}
\end{table}

	We first classify the complexity of the problems for each fixed size $\doms$ of a domain in Section~\ref{sec:doms}; recall that $\doms$ corresponds to the number of colors in \CR{(L)}.
	Together with known results, our results give interesting boundaries of (in)tractability as summarized in \tablename~\ref{tab:result_k}.
	In particular, our results unravel the boundaries with respect to $\doms$ for \CSR{2} and \HR{L}.
	The other interesting contrast we show is the boundary between \CSR{2} and \HR{L} for $\doms=3$.
	This might be caused by the difference of variety of constraints; in \HR{L}, a constraint is unique for each pair of lists.
	
	We then extend Theorem~\ref{the:k2easy}, which states that \CSR{2} can be solved in polynomial time if $\doms=2$.
	More precisely, we give a fixed-parameter algorithm with respect to $\doms+\NB$, where $\NB$ is the number of vertices whose lists have size more than two.
	Notice that $\doms=2$ implies that $\NB=0$, and hence our algorithm generalizes Theorem~\ref{the:k2easy}.
	Moreover, this result gives the boundary of the complexity of \CSR{2} between $\doms=3$ and $\doms=2$.
	We also investigate the parameterized complexity when parameterized only by $\NB$.
	The results are summarized in \tablename~\ref{tab:result_NB}.

\paragraph*{Graph parameters}

\begin{table}[t]
	\centering
	\caption{Computational complexity for graphs with pathwidth at most two.}
	\medskip
	\label{tab:result_pw}
	\begin{tabular}{|l|l|l|}
		\hline
		& $\PW=2$ & $\PW=1$ \\ \hline
		\CSR{0}	& PSPACE-c.		& PSPACE-c.		\\ \hline
		\CSR{3}	& PSPACE-c.		& PSPACE-c.		\\ \hline
		\CSR{2}	& PSPACE-c.		& PSPACE-c.		\\ \hline
		\HR{L}	& PSPACE-c.		& PSPACE-c.~[Thm.~\ref{the:pathhard}]	\\ \hline
		\CR{L}	& PSPACE-c.~\cite{HIZ14,Wro18}	& P~\cite{HIZ14}	\\ \hline
		\HR{}		& PSPACE-c.~\cite{Wro18}		& P~\cite{Wro18}	\\ \hline
		\CR{}		& P~\cite{HIZ17c}	& P	\\ \hline
	\end{tabular}
\end{table}

\begin{table}[t]
	\centering
	\caption{Parameterized complexity with respect to $\doms$ plus graph parameters.}
	\medskip
	\label{tab:result_graph}
	\begin{tabular}{|l|l|l|l|l|}
		\hline
		Parameter & $\doms+\MW$	& $\doms+\TD$	& $\doms+\VC$ & $\doms+\BW$ \\ \hline
		\CSR{0}	& PSPACE-c.		& FPT~[Thm.~\ref{the:tdeasy}]	& FPT~[Thms.~\ref{the:vc}, \ref{the:vcfast}]	& PSPACE-c.		\\ \hline
		\CSR{3}	& PSPACE-c.		& FPT		& FPT		& PSPACE-c.	\\ \hline
		\CSR{2}	& PSPACE-c.~[Cor.~\ref{cor:mwhard}]	& FPT		& FPT		& PSPACE-c.		\\ \hline
		\HR{L}	& FPT~[Thm.~\ref{the:mweasy}]		& FPT		& FPT		& PSPACE-c.			\\ \hline
		\CR{L}	& FPT~\cite{HIZ18}	& FPT		& FPT		& PSPACE-c.	\\ \hline
		\HR{}		& FPT	& FPT~\cite{Wro18}	& FPT		& PSPACE-c.		\\ \hline
		\CR{}		& FPT	& FPT		& FPT		& PSPACE-c.~\cite{Wro18}	\\ \hline
	\end{tabular}
\end{table}

	As mentioned in the previous subsection, \CR{}, \CR{L} and \HR{} have been studied from the viewpoint of graph parameters.
	In this paper, we extend the notion of graph parameters to hypergraphs by taking a ``primal graph.''
	The \emph{primal graph} $\Prm{G}$ of a hypergraph $G$ is a graph such that $V(\Prm{G})=V(G)$ and two distinct vertices are connected by an edge if they are contained in the same hyperedge of $G$.
	Then, we define any graph parameter of a hypergraph $G$ as the parameter of its primal graph $\Prm{G}$.\footnote{
		For example, when we refer to the treewidth of a hypergraph $G$, it means the treewidth of its primal graph $\Prm{G}$.
		Note that $\Prm{G}=G$ if $G$ is $2$-uniform.}
	Then we can draw \tablename s~\ref{tab:result_pw}	and \ref{tab:result_graph} from this viewpoint; throughout the paper, $\MW$, $\PW$, $\TD$ and $\VC$ are the modular-width, the pathwidth, the tree-depth and the size of a minimum vertex cover of a given hypergraph, respectively.
	The relationships between graph parameters are summarized in \figurename~\ref{fig:relation}(b); note that tractability (resp., intractability) result propagates downward (resp., upward).
	We extend several known fixed-parameter algorithms to \HR{2} ($\doms+\MW$), and \CSR{0} ($\doms+\TD$ and $\doms+\VC$).
	On the other hand, we show that the first one cannot be extended to \CSR{2}.
	We also note that a fixed-parameter tractability of \CSR{0} parameterized by $\doms+\VC$ can be obtained as a corollary of Theorem~\ref{the:tdeasy}.
	However, Theorem~\ref{the:vcfast} gives a faster algorithm and Theorem~\ref{the:vc} gives an algorithm for the shortest variant.
	
\paragraph*{Boundary of time comlexity}

	In Section~\ref{sec:boundary}, we explore the boundary of time complexity of \CSR{0}.
	We first give an algorithm for \CSR{0} running in time $O^*(\doms^{O(n)})$, where $n$ is the number of vertices of a given hypergraph, and hence \CSR{0} is in the class XP when parameterized by $n$.
	On the other hand, we prove that \HR{} is $W[1]$-hard when parameterized by $n$.
	Furthermore, we prove that \CSR{2} cannot be solved in time $\runtime{(\doms+n)}{\doms+n}$ under the exponential time hypothesis (ETH).
	This lower bound matches the running time shown in Theorems~\ref{the:exact}, \ref{the:vcfast} and \ref{the:nb}.	
	
	\medskip
	We move several proofs to Appendices.

\subsection{Preliminary}
\label{sec:pre}
	
	Let $G$ be a hypergraph, and let $v\in V(G)$ be a vertex.
	We denote by $\Neigh{G}{v}$ the set $\{w \colon w\ne v,\ \exists X\in E(G),\ \{v,w\}\subseteq X\}$ of vertices which are \emph{adjacent} to $v$.
	For a vertex subset $V^\prime \subseteq V(G)$, we denote $\Neigh{G}{V^\prime}:=\bigcup_{v\in V^\prime} \Neigh{G}{v} \setminus V^\prime$.
	
	Two hypergraphs $G$ and $G^\prime$ are \emph{isomorphic} if there exist two bijections $\isov \colon V(G) \to V(G^\prime)$ and $\isoe \colon E(G) \to E(G^\prime)$ such that $\isoe(X)=\{\isov(v_1),\isov(v_2),\ldots,\isov(v_\ari)\}\in E(G^\prime)$ holds for each hyperedge $X=\{v_1,v_2,\ldots,v_\ari\}\in E(G)$.
	(See \figurename~\ref{fig:isomorphism}.)
	For a hypergraph $G$ and a vertex subset $V^\prime \subseteq V(G)$, we define the \emph{subhypergraph} of $G$ \emph{induced by} $V^\prime$ as the hypergraph $G^\prime$ such that $V(G^\prime)=V^\prime$ and $E(G^\prime)=\{X\cap V^\prime \colon X\in E(G), X\cap V^\prime \ne \emptyset\}$.
	We denote by $G[V^\prime]$ the subhypergraph of $G$ induced by $V^\prime$ for any vertex subset $V^\prime$.
	(See \figurename~\ref{fig:subhypergraph}.)
	We use the notation $G\setminus V^\prime$ to denote $G[V(G)\setminus V^\prime]$.
	
	\begin{figure}[t]
		\begin{center}
			\includegraphics{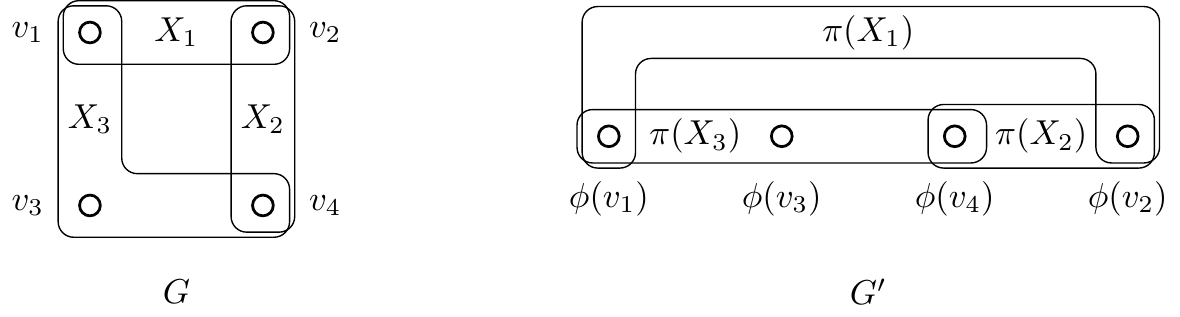}
		\end{center}
		\caption{Two isomorphic hypergraphs $G$ and $G^\prime$ under the bijections $\isov$ and $\isoe$.}
		\label{fig:isomorphism}
	\end{figure}
	\begin{figure}[t]
		\begin{center}
			\includegraphics{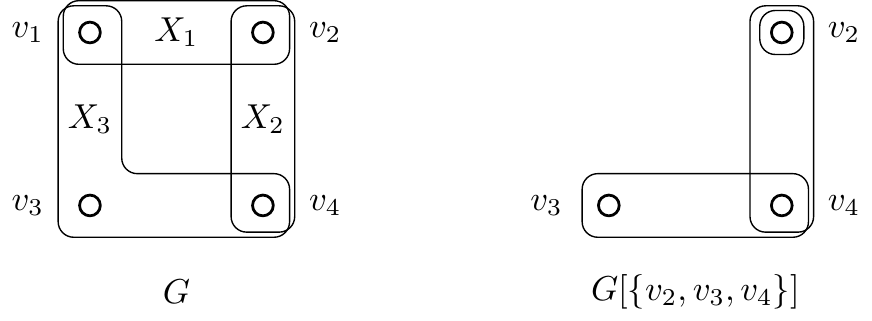}
		\end{center}
		\caption{A graph $G$ and the subhypergraph $G[\{v_2,v_3,v_4\}]$ induced by $\{v_2,v_3,v_4\}$.}
		\label{fig:subhypergraph}
	\end{figure}

	We say that two mappings $\phi \in B^{A}$ and $\phi^\prime \in B^{A^\prime}$ are \emph{compatible} if $\rest{\phi}{A\cap A^{\prime \prime}}=\rest{\phi^\prime}{A\cap A^{\prime \prime}}$ holds.

	
	Let $(G,\Dom,\Ecns)$ be an instance of \CSP{0}.
	A \emph{Boolean vertex} is a vertex $v\in V(G)$ with $|\Vcns(v)|\le 2$, and a \emph{non-Boolean vertex} is a vertex $v\in V(G)$ with $|\Vcns(v)|>2$.
	Let $X$ and $X^\prime$ be hyperedges in $E(G)$ such that $|X|=|X^\prime|$.
	For a bijection $\isov \colon X\to X^\prime$, we denote by $\repcns{\Ecns}{\isov}(X)$ the set $\{\mapg \circ \isov^{-1} \colon \mapg \in \Ecns(X)\}\subseteq \Dom^{X^\prime}$ of mappings from $X^\prime$ to $\Dom$,
	where $\circ$ means the composition of mappings.
	Intuitively, $\repcns{\Ecns}{\isov}(X)$ is a ``translation'' of $\Ecns(X)$ into a constraint of $X^\prime$ via a bijection $\isov$.
	For example, assume that $\Ecns(\{v_1,v_2,v_3\})$ contains a mapping $\mapg=(1,2,3)$.
	If a bijection $\phi \colon \{v_1,v_2,v_3\} \to \{u_1,u_2,u_3\}$ maps $v_1,v_2,v_3$ to $u_2,u_1,u_3$, respectively, 
	then $\repcns{\Ecns}{\isov}(\{v_1,v_2,v_3\})$ contains a mapping $\mapg^\prime \colon \{u_1,u_2,u_3\} \to \{1,2,3\}$ such that $(\mapg^\prime(u_1),\mapg^\prime(u_2),\mapg^\prime(u_3))=(\mapg \circ \isov^{-1}(u_1),\mapg \circ \isov^{-1}(u_2),\mapg \circ \isov^{-1}(u_3))=(\mapg(v_2),\mapg(v_1),\mapg(v_3))=(2,1,3)$.
	
	Let $(G,\Dom,\Ecns)$ be an instance of \CSP{2}.
	Without loss of generality, we assume that $G$ is connected, $|V(G)|\ge 2$, and $\Dom=\{0,1\}$.
	Moreover, we can assume that $G$ is $2$-uniform as follows.
	If $G$ contains a size-one hyperedge $\{v\}$, there must exist a size-two hyperedge (i.e., an edge) $vw\in E(G)$ from the assumption.
	Then, we remove $\{v\}$ from $E(G)$ and replace $\Ecns(vw)$ with the set of all solutions satisfying $\Ecns(\{v\})$ and $\Ecns(vw)$.
	Note that this modification does not change the set of solutions and the primal graph.

\section{Computational complexity with respect to $\boldsymbol{\doms}$}
\label{sec:doms}

	In this section, we classify the complexity of the problems for each fixed size $\doms$ of a domain.
	\begin{theorem}
		\label{the:k3hard}
		\CSR{2} is PSPACE-complete for bipartite planar graphs even if $\doms=3$.
	\end{theorem}
	\begin{proof}
		We give a polynomial-time reduction from \CR{L} to \CSR{2}.
		It is known that \CR{L} is PSPACE-complete for bipartite planar graphs even if each list has size at most three~\cite{BC09}.
		Let $\Inst=(G,\Dom,\Ecns,\mapf_\ini,\mapf_\tar)$ be such an instance of \CR{L}, and let $\Vcns$ is a list assignment.
		Without loss of generality, we assume that $\Dom=\{1,2,3,4\}$.
		We then construct an instance $(G,\{1,2,3\},\Ecns^\prime,\mapf_\ini^\prime,\mapf_\tar^\prime)$ of \CSR{2} as follows.
		The idea is to simply replace a value $4$ with some value from $\{1,2,3\}$ for each vertex without changing the graph $G$.
		Let $v\in V(G)$ be a vertex such that $4\in \Vcns(v)$.
		Since $|\Vcns(v)|\le 3$, there exists a value $i$ in $\{1,2,3\}\setminus \Vcns(v)$.
		Let $\pi \colon \Dom \to \Dom$ be a permutation such that $\pi(i)=4$, $\pi(4)=i$ and $\pi(j)=j$ for each $j\notin \{i,4\}$, and we update $(\Vcns,\Ecns,\mapf_\ini,\mapf_\tar)$ as follows:
		\begin{itemize}
			\item $\Vcns(v):=\Vcns(v)\setminus \{4\} \cup \{i\}$;
			\item $\mapg:=(\pi(\mapg(v)),\mapg(w))$ for each neighbor $w$ of $v$ and each mapping $\mapg \in \Ecns(vw)$; and
			\item $\mapf_\iot(v):=\pi(\mapf_\iot(v))$ for each $\iot \in \{\ini,\tar\}$.
		\end{itemize}
		We repeat this operation until there is no vertex $v$ such that $4\in \Vcns(v)$, and let $\Inst^\prime=(G,\{1,2,3\},\Ecns^\prime,\mapf_\ini^\prime,\mapf_\tar^\prime)$ be the resulting instance.
		Observe that the construction can be done in polynomial time.
		
		Because we only replace values, $\Inst^\prime$ is a valid instance of \CSR{2} which is essentially equivalent to $\Inst$.
		Moreover, $G$ is bipartite planar and the domain has only three values, and thus this completes the proof.
	\end{proof}
	
	In contrast to Theorem~\ref{the:k3hard}, there exist polynomial-time algorithms for more restricted cases.
	We first show that the problem becomes tractable when restricted to \HR{L} and $\doms=3$.
	\begin{theorem}
		\label{the:k3easy}	
		\HR{L} can be solved in polynomial time if $\doms=3$.
	\end{theorem}
	\begin{proof}
		Let $\Inst=(G,\Dom,\Ecns,\mapf_\ini,\mapf_\tar)$ be a given instance of \HR{L} such that $H$ is the \underlying{} graph with $|V(H)|\le |\Dom|=3$ and $\Vcns$ is a list assignment.
		We assume without loss of generality that $G$ is connected and has at least two vertices.
		Since $G$ is connected, for any homomorphism $\mapf$ from $G$ to $H$, there exists exactly one connected component $C$ of $H$ such that $\mapf(v)\in V(C)$ holds for every vertex $v\in V(G)$.
		Moreover, for any two homomorphisms $\mapf$ and $\mapf^\prime$ from $G$ to $H$ which correspond to different connected components, $|\diff{\mapf}{\mapf^\prime}|=|V(G)|\ge 2$ holds; and hence they are not adjacent in the solution graph.
		Because the relation of reconfigurability between homomorphisms is transitive, $C_\ini=C_\tar$ holds if $\mapf_\ini$ and $\mapf_\tar$ are reconfigurable, where $C_\ini$ and $C_\tar$ are connected components of $H$ corresponding to $\mapf_\ini$ and $\mapf_\tar$, respectively.
		Thus, we can assume that $C_\ini=C_\tar$, and let $H:=C_\ini$ and $\Dom:=V(C_\ini)$.
		If $H$ is complete, $\Inst$ is also an instance of \CR{L} with $|\Dom|\le 3$, which is solvable in polynomial time~\cite{CHJ11}.
		Otherwise, $H$ is a path $(\{1,2,3\},\{12,23\})$ of length two.
		Let $V_\ini$ (resp. $V_\tar$) be the set of all vertices $v\in V(G)$ with $\mapf_\ini(v)\in \{1,3\}$ ($\mapf_\ini(v)\in \{1,3\}$).
		We now claim that $\Inst$ is a yes-instance if and only if $V_\ini=V_\tar$, which can be checked in polynomial time.
		
		If $V_\ini=V_\tar$, then $\diff{\mapf_\ini}{\mapf_\tar}\subseteq V_\ini$ holds.
		Because $H$ contains no edge between $1$ and $3$ and $\mapf_\ini$ is a homomorphism from $G$ to $H$, $V_\ini$ must be an independent set of $G$.
		Therefore, we can independently change the value of each vertex in $\diff{\mapf_\ini}{\mapf_\tar}\subseteq V_\ini$ to obtain $\mapf_\tar$; and hence $\Inst$ is a yes-instance.
		
		We next assume that $\Inst$ is a yes-instance but $V_\ini \ne V_\tar$.
		Then, there exist two consecutive homomorphisms $\mapf$ and $\mapf^\prime$ in the reconfiguration sequence such that $(\mapf(v),\mapf^\prime(v))\in \{(1,2),(2,1),(3,2),(2,3)\}$ holds, where $v$ is the unique vertex in $\diff{\mapf}{\mapf^\prime}$;
		that is, $\mapf^\prime$ is obtained from $\mapf$ by changing the value of $v$ along an edge of $H$.
		Since $G$ is connected and has at least two vertices, $v$ has at least one neighbor $w$ in $G$.
		Because $\mapf$ is a homomorphism, $\mapf(v)=2$ if and only if $\mapf(w)\ne 2$.
		Similarly, because $\mapf^\prime$ is a homomorphism, $\mapf^\prime(v)=2$ if and only if $\mapf^\prime(w)=\mapf(w)\ne 2$.
		From the definition of $v$, $\mapf(v)=2$ if and only if $\mapf^\prime(v)\ne 2$.
		We thus have that $\mapf(w)\ne 2$ if and only if $\mapf(w)=2$, which is a contradiction.
		Therefore, $V_\ini \ne V_\tar$ if $\Inst$ is a yes-instance.
	\end{proof}
	
	We next show that \CSR{2} becomes tractable if $\doms$ is reduced from three to two.
	\begin{theorem}
		\label{the:k2easy}
		\CSR{2} can be solved in polynomial time if $\doms=2$.
	\end{theorem}
	\begin{proof}
		We reduce the problem to \textsc{bijunctive} \BCSR{}, which is solvable in polynomial time~\cite{GKMP09}.
		\textsc{Bijunctive} \BCSR{} is a special case of \BCSR{} where $\Dom=\{0,1\}$ and there exists a $2$-CNF formula $\phi(v_1,\ldots,v_r)$ such that $\Ecns(\{v_1\ldots,v_r\})$ is exactly the set of all satisfying assignments of $\phi$ for every hyperedge $\{v_1\ldots,v_r\}\in E(G)$.
		Let $\Inst=(G,\Dom,\Ecns,\mapf_\ini,\mapf_\tar)$ be a given instance of \CSR{2} where $G$ is a graph and $\Dom=\{0,1\}$.
		We now show that for every edge $vw\in E(G)$ there exists a $2$-CNF formula $\phi(v,w)$ such that $\Ecns(vw)$ is exactly the set of all satisfying assignments of $\phi$.
		For each $i\in \Dom$ and $u\in \{v,w\}$, we denote by $u^i$ a literal $u$ if $i=0$ or $\bar{u}$ if $i=1$.
		Then we define a $2$-CNF formula $\phi(v,w)$ as follows:
		\[
		\phi(v,w)=\bigwedge_{(a,b)\in \Dom^2 \setminus \Ecns(vw)} (v^a \vee w^b).
		\]
		Notice that a clause $(v^a \vee w^b)$ corresponds to a set $\Dom^2 \setminus \{(a,b)\}$.
		Therefore, $\phi(v,w)$ corresponds to the set
		\[
		\bigcap_{(a,b)\in \Dom^2 \setminus \Ecns(vw)} \Dom^2 \setminus \{(a,b)\}=\Dom^2 \setminus (\Dom^2 \setminus \Ecns(vw))=\Ecns(vw)
		\]
		as required.
	\end{proof}

\section{Boundary of time complexity}
\label{sec:boundary}

	In this section, we explore the boundary of time complexity of the problem.
	We first give a simple exact algorithm as follows.
	\begin{theorem}
		\label{the:exact}
		\CSR{0} can be solved in time $O^*(\doms^{O(n)})$, and hence \CSR{0} is fixed-parameter tractable when parameterized by $\doms+n$ and in XP when parameterized by $n$, where $n$ is the number of vertices of a given hypergraph.
	\end{theorem}
	\begin{proof}
		Our algorithm explicitly construct the solution graph and then check the connectivity between two given solutions.
		The solution graph has at most $\doms^n$ vertices and can be constructed in time $O^*(\doms^{O(n)})$.
		The connectivity can be checked in time polynomial in the size of the solution graph by a simple breadth-first search.
		Therefore, our algorithm runs in time $O^*(\doms^{O(n)})$.
	\end{proof}
	
	On the other hand, the following theorem implies that a fixed-parameter algorithm parameterized only by $n$ is unlikely to exist.
	\begin{theorem}
		\label{the:W1hard}
		\HR{} is $W[1]$-hard when parameterized by $n$.
	\end{theorem}
	\begin{proof}
		We give a parameterized reduction from \textsc{labeled clique reconfiguration}, which is defined as follows.
		Let $G^\prime$ be a simple graph and let $\tk$ be a positive integer.
		A $\tk$\emph{-labeled clique} ($\tk$\emph{-LC} for short) of $G^\prime$ is a vector $(u_1,u_2,\ldots,u_{\tk})$ consisting of $\tk$ distinct vertices $u_1,u_2,\ldots,u_{\tk}\in V(G^\prime)$ which form a clique.
		A $\tk$\emph{-labeled clique graph} $\Cls{\tk}{G^\prime}$ is a graph such that $V(\Cls{\tk}{G^\prime})$ is a set of all $\tk$-LCs of $G^\prime$, and two $\tk$-labeled cliques $C$ and $C^\prime$ are adjacent if and only if they differ on exactly one component.
		Then, \textsc{labeled clique reconfiguration} asks for a given graph $G^\prime$, an integer $\tk$, two $\tk$-LCs $C_\ini$ and $C_\tar$, whether there exists a walk between $C_\ini$ and $C_\tar$ in $\Cls{\tk}{G^\prime}$ or not.
		It is known that \textsc{labeled clique reconfiguration} is $W[1]$-hard when parameterized by $\tk$~\cite{IKOSUY14}.\footnote{
			Although they actually show the similar result for (unlabeled) \textsc{independent set reconfiguration}, the proof can be applied to \textsc{labeled clique reconfiguration}.}
		
		We now construct an instance $(G,\Dom,\Ecns,\mapf_\ini,\mapf_\tar)$ of \HR{} corresponding to an instance $(G^\prime,\tk,C_\ini,C_\tar)$ of \textsc{labeled clique reconfiguration}.
		Let $G$ be a complete graph $K_\tk$ with $\tk$ vertices $\{v_1,v_2,\ldots,v_\tk\}$, and let $\Dom=V(G^\prime)$.
		We define constraints for edges so that $G^\prime$ is an \underlying{} graph; that is, for each $v_i v_j \in E(G)$, we define $\Ecns(v_i v_j)=\{(u_p,u_q) \colon u_p u_q \in E(G^\prime)\}$.
		Observe that $(G,\Dom,\Ecns)$ is an instance of \textsc{homomorphism} with $\tk$ vertices.
		The remaining components, two solutions $\mapf_\ini$ and $\mapf_\tar$, are defined as follows.
		For any $\tk$-LC $C=(u_1,u_2,\ldots,u_{\tk})$ of $G^\prime$, we define $\phi_C$ be a mapping such that $\phi_C(v_i)=u_i$ for each $i\in \{1,2,\ldots,\tk\}$.
		Since $\phi_C(v_i) \phi_C(v_j)=u_i u_j\in E(G^\prime)$ holds for each distinct $i,j\in \{1,2,\ldots,\tk\}$, $\phi_C$ is a solution of $(G,\Dom,\Ecns)$.
		Thus, let $\mapf_\ini=\phi_{C_\ini}$ and $\mapf_\tar=\phi_{C_\tar}$.
		This completes the construction of $\Inst=(G,\Dom,\Ecns,\mapf_\ini,\mapf_\tar)$.
		Finally, $(G^\prime,\tk,C_\ini,C_\tar)$ is a yes-instance if and only if $\Inst$ is, because $\Cls{\tk}{G^\prime}$ and $\Sol{\Inst}$ are isomorphic under a mapping $\phi$.
	\end{proof}
	The following also follows as a corollary.
	\begin{corollary}
		\label{cor:W1hard}
		\HR{} is $W[1]$-hard when parameterized by $p$, where $p$ is any parameter which is polynomially bounded in $n$.
	\end{corollary}

	We finally give a lower bound on the computation time, which matches the running times of Theorem~\ref{the:exact} and two theorems which will be shown later.
	\begin{theorem}[*]
		\label{the:LB}
		Under ETH, there exists no algorithm solving \CSR{2} in time $\runtime{(\doms+n)}{\doms+n}$. 
	\end{theorem}

\section{PSPACE-completeness for graphs with bounded parameters}
\label{sec:gpara}

	In this section, we show the PSPACE-completenesses of \HR{L} and \CSR{2} for graphs with bounded parameter.

	\begin{theorem}[*]
		\label{the:pathhard}
		\HR{L} is PSPACE-complete even if $\doms=O(1)$ for paths.
	\end{theorem}

	\begin{corollary}
		\label{cor:mwhard}
		\CSR{2} is PSPACE-complete even if $\doms=3$ for complete graphs.
	\end{corollary}
	\begin{proof}
		Let $(G,\Dom,\Ecns,\mapf_\ini,\mapf_\tar)$ be an instance of \CSR{2} constructed in Theorem~\ref{the:k3hard}.
		We then add an edge between every non-adjacent pair, and give a trivial constraint $\Ecns(vw)=\Dom^2$ to every added edge $vw$.
		Notice that this modification does not change the solution graph, and thus the reconfigurability.
	\end{proof}

\section{Fixed-parameter algorithm with respect to graph parameters}
\label{sec:kernelFPT}

	We give the following theorems in this section.
	\begin{theorem}[*]
		\label{the:mweasy}
		\HR{L} is fixed-parameter tractable when parameterized by $\doms+\MW$.
	\end{theorem}
	\begin{theorem}[*]
		\label{the:tdeasy}
		\CSR{0} is fixed-parameter tractable when parameterized by $\doms+\TD$.
	\end{theorem}

\subsection{Reduction rule}
\label{sec:reduce}
	
	In order to prove the above theorems, we give fixed-parameter algorithms which are based on the concept of ``kernelization''.
	That is, we compute from the given instance into another instance whose size depends only on the parameter.
	After that, we can solve the problem by using Theorem~\ref{the:exact}.
	
	In this subsection, we show some useful lemma, which compresses an input hypergraph into a smaller hypergraph with keeping the reconfigurability.
	We note that this is the extension of the lemma given in \cite{HIZ18} to obtain a fixed-parameter algorithm for \CR{L} parameterized by $\doms+\MW$.
	The main idea is to ``identify'' two subgraphs which behave in the same way with respect to the reconfigurability.
	
	We now formally characterize such subhypergraphs and explain how to identify them.
	Let $\Inst=(G,\Dom,\Ecns,\mapf_\ini,\mapf_\tar)$ be an instance of \CSR{0}.
	For each vertex $v\in V(G)$, we define $\asgn{v}$ as a pair $(\mapf_{\ini}(v),\mapf_{\tar}(v))$ consisting of the initial and the target value assignments of $v$.
	Let $V_1$ and $V_2$ be two non-empty vertex subsets of $G$ such that $|V_1|=|V_2|$, and $V_1\cap V_2=\emptyset$.
	Assume that $\Neigh{G}{V_1}=\Neigh{G}{V_2}=W$.
	Let $H_1=G[V_1]$, $H_2=G[V_2]$, $H_1^\prime=G[V_1\cup W]$ and $H_2^\prime=G[V_2\cup W]$.
	\begin{definition}
		\label{def:identical}
		Two induced subhypergraphs $H_1$ and $H_2$ are \emph{identical} if there exist two bijections $\isov \colon V(H_1^\prime) \to V(H_2^\prime)$ and $\isoe \colon E(H_1^\prime) \to E(H_2^\prime)$ which satisfy the following four conditions:
		\begin{enumerate}[(1)]
			\item $H_1^\prime$ and $H_2^\prime$ are isomorphic under $\isov$ and $\isoe$.
			\item for every vertex $v\in W$, $\isov(v)=v$;
			\item for every vertex $v\in V_1$, $\asgn{v}=\asgn{\isov(v)}$, that is, $\mapf_{\ini}(v)=\mapf_{\ini}(\isov(v))$ and $\mapf_{\tar}(v)=\mapf_{\tar}(\isov(v))$; and
			\item for every hyperedge $X\in E(H_1)$, $\Ecns(\isoe(X))=\repcns{\Ecns}{\widehat{\isov}}(X)$, where $\widehat{\isov}=\rest{\phi}{X}$.
		\end{enumerate}
	\end{definition}
	See \figurename~\ref{fig:identical} for an example.

\begin{figure}[t]
	\begin{center}
		\includegraphics{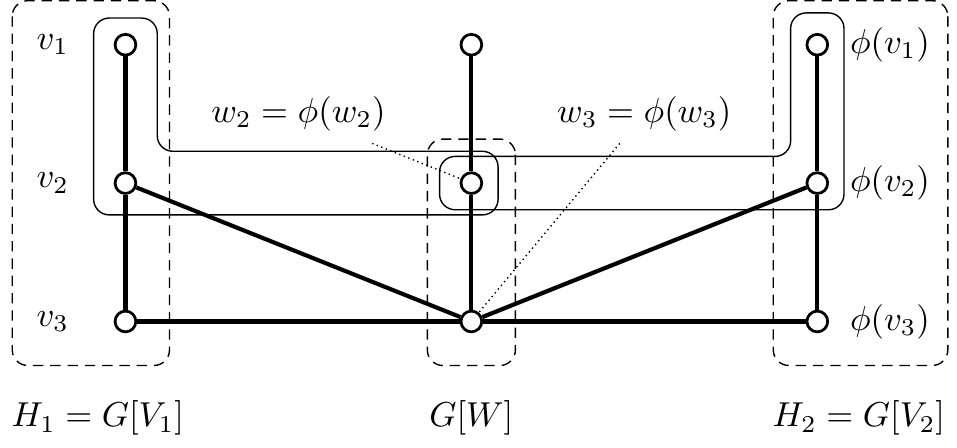}
	\end{center}
	\caption{An example of two subhypergraphs $H_1$ and $H_2$ of $G$ which satisfies the conditions (1) and (2).
		We draw each hyperedge of size two as a solid line, and omit the bijection $\isoe \colon E(H_1^\prime) \to E(H_2^\prime)$ since it is uniquely defined from $\isov \colon V(H_1^\prime) \to V(H_2^\prime)$.
		If $\asgn{0}$ and $\Ecns$ satisfy the conditions (3) and (4), $H_1$ and $H_2$ are identical.}
	\label{fig:identical}
\end{figure}

	\medskip
	We next define another instance $\Inst^\prime=(G^\prime,\Dom,\Ecns^\prime,\mapf_\ini^\prime,\mapf_\tar^\prime)$ as follows:
	\begin{itemize}
		\item $G^\prime=G \setminus V_2$;
		\item $\mapf_\ini^\prime=\rest{\mapf_\ini}{V(G^\prime)}$ and $\mapf_\ini^\prime=\rest{\mapf_\ini}{V(G^\prime)}$; and
		\item  for each $X^\prime \in E(G^\prime)$, $\Ecns^\prime(X^\prime)=\bigcap_{X\in E^\prime}\subcns{X}$, 
		where $E^\prime=\{X\in E(G) \colon X\setminus V_2=X^\prime \}$ and $\subcns{X}=\{\rest{\mapg}{X^\prime} \colon \mapg \in \Ecns(X)\}$.
	\end{itemize}
	Intuitively, $\Inst^\prime$ is obtained by restricting all components (hypergraphs, mappings in constraints, and two solutions) of $\Inst$ on $V(G)\setminus V_2$.
	We say that $\Inst^\prime$ is obtained from $\Inst$ by \emph{identifying} $H_1$ with $H_2$.
	
	\medskip
	Then, we have the following lemma, which says that $\Inst$ and $\Inst^\prime$ are equivalent with respect to the feasibility.
	\begin{lemma}[*]
		\label{lem:reduce_pre}
		Let $\mapf^\prime \colon V(G^\prime) \to \Dom$ be a mapping from $V(G^\prime)$ to $\Dom$.
		Then, $\mapf^\prime$ is a solution for $(G^\prime,\Dom,\Ecns^\prime)$ if and only if there exists a solution $\mapf$ for $(G,\Dom,\Ecns)$ such that $\mapf^\prime=\rest{\mapf}{V(G^\prime)}$.
	\end{lemma}
	Notice that Lemma~\ref{lem:reduce_pre} ensures that the ``restricted'' instance $\Inst^\prime$ is a valid instance of \CSR{0}.
	
	\medskip
	We now give the following key lemma, which says that $\Inst$ and $\Inst^\prime$ are equivalent with respect to even the reconfigurability.
	\begin{lemma}[*, Reduction rule]
		\label{lem:reduce}
		Let $\Inst$ and $\Inst^\prime$ be instances of \CSR{0} defined as above.
		Then, $\Inst^\prime$ is a yes-instance if and only if $\Inst$ is.
	\end{lemma}

\section{Vertex cover}
\label{sec:VC}

	In this section, we consider the size $\VC$ of a minimum vertex cover.
	Note that Theorem~\ref{the:tdeasy} implies \CSR{0} is fixed-parameter tractable when parameterized by $\doms+\VC$.
	We strengthen it as follows.
	\begin{theorem}[*]
		\label{the:vc}
		The shortest variant of \CSR{0} is fixed-parameter tractable when parameterized by $\doms+\VC$.
	\end{theorem}
	\begin{theorem}[*]
		\label{the:vcfast}
		There exists a fixed-parameter algorithm for \CSR{0} parameterized by $\doms+\VC$ which runs in time $O^*(\doms^{O(\VC)})$.
	\end{theorem}

\subsection{Discussions}

	We conclude this section by discussing hitting sets on hypergraphs, which is a well-known generalization of vertex covers on graphs.
	Although a hitting set of a $2$-uniform hypergraph is equivalent to a vertex cover of the graph, such an equivalence does not hold for general hypergraphs.
	Thus, it is worth considering the complexity of \CSR{0} with respect to the size of a hitting set of a given hypergraph.
	We have the following theorem, which implies that a fixed-parameter algorithm for \CSR{0} is unlikely to exist when parameterized by the size of a hitting set plus $\doms$.
	\begin{theorem}
		\label{the:hitting}
		\CSR{3} is PSPACE-complete even for hypergraphs with a hitting set of size one and $\doms=3$.
	\end{theorem}
	\begin{proof}
		Let $\Inst=(G,\Dom,\Ecns,\mapf_\ini,\mapf_\tar)$ be an instance constructed in Theorem~\ref{the:k3hard}.
		Briefly speaking, we add a new vertex $u$ to $G$, include it in every edge $vw\in E(G)$, and modify the constraints so that a value assignment to $u$ does not affect any other vertices.
		More precisely, we construct a new instance $\Inst^\prime=(G^\prime,\Dom,\Ecns^\prime,\mapf^\prime_\ini,\mapf^\prime_\tar)$ as follows.
		Let $V(G^\prime):=V(G)\cup \{u\}$ and $E(G^\prime):=\{\{u,v,w\} \colon vw \in E(G)\}$, where $u$ is a new vertex which is not in $G$.
		Then, $G^\prime$ has a hitting set $\{u\}$ of size one.
		For each hyperedge $\{u,v,w\}\in E(G^\prime)$, we let $\Ecns^\prime(\{u,v,w\}):=\{1\} \times \Ecns(vw)$.
		We finally extend $\mapf_\ini$ (resp., $\mapf_\tar$) to $\mapf^\prime_\ini$ (resp., $\mapf^\prime_\tar$) by setting $\mapf^\prime_\ini(u)=1$ (resp., $\mapf^\prime_\tar(u)=1$).
		This completes the construction.
		 $\Inst^\prime$ is equivalent to $\Inst$ by ignoring a value assignment to $u$.
	\end{proof}

\section{Extension of Theorem~\ref{the:k2easy}}
\label{sec:otherFPT}

	In this section, we extend Theorem~\ref{the:k2easy} to a more general situation.
	\begin{theorem}[*]
		\label{the:nb}
		There exists a fixed-parameter algorithm for \CSR{2} parameterized by $\doms+\NB$ which runs in time $O^*(\doms^{O(\NB)})$.
	\end{theorem}
	
\section{Conclusion}
	
	In this paper, we studied \CSR{0} and its special cases from the viewpoints of polynomial-time solvability and parameterized complexity, 
	and gave several interesting boundaries of tractable and intractable cases with respect to the size $\doms$ of a domain (\tablename~\ref{tab:result_k}), the size of lists (\tablename~\ref{tab:result_NB}), and the structure of an input hypergraph (\tablename s~\ref{tab:result_pw} and \ref{tab:result_graph}).
	Moreover, we gave some exact algorithms, and a lower bound of the computation time that meets the running times of our algorithms under ETH.
	
	We leave as an open question whether some of our algorithms can be extended for the shortest variant or not.
	In particular, the extensions of the fixed-parameter algorithms parameterized by $\doms+\MW$ (Theorem~\ref{the:mweasy}) and $\doms+\TD$ (Theorem~\ref{the:tdeasy}), respectively, are interesting.
	We note that the reduction rule (Lemma~\ref{lem:reduce}) does not preserve the shortest length of a reconfiguration sequence since we just remove vertices.

\newpage
\appendix
	
\section{Proof of Theorem~\ref{the:LB}}

	We give a polynomial-time reduction from \kkc{}, which is defined as follows.
	An instance of \kkc{} is a graph $H$ with the vertex set $\{\kkv{i}{p} \colon 1 \le i,p \le \kk \}$; we denote $\kkgp{i}=\{\kkv{i}{p} \colon 1 \le p \le \kk \}$ for each $i\in \{1,2,\ldots,\kk\}$.
	Then, the problem asks whether there exists a clique $Q\subseteq V(H)$ such that $|Q\cap \kkgp{i}|=1$ for every $i\in \{1,2,\ldots,\kk\}$.
	It is known that there exists no algorithm solving \kkc{} in time $\runtime{\kk}{\kk}$ under ETH~\cite[Theorem 14.12]{ETHlb15}.
	We will present a polynomial-time transformation from an instance $H$ of \kkc{} to an instance $\Inst$ of \CSR{2} such that
	\begin{itemize}
		\item a graph has $\kk+2$ vertices and a domain has $\kk+1$ values; and 
		\item $H$ is a yes-instance if and only if $\Inst$ is.
	\end{itemize}
	If there exists an algorithm $A$ solving \CSR{2} in time $\runtime{(\doms+n)}{\doms+n}$, an execution of $A$ for the transformed instance $\Inst$ yields an algorithm solving \kkc{} in time
	\[
	\runtime{(\doms+n)}{\doms+n}=\runtime{(2\kk+3)}{2\kk+3}=\runtime{(\kk^{2})}{2\kk+3}=\runtime{\kk}{\kk}.
	\]
	
	Before constructing $\Inst$, we first reformulate \kkc{} as a special case of $2$\textsc{-ary constraint satisfiability} by a similar idea of the proof of Theorem~\ref{the:W1hard}.
	Let $G^\prime$ be a complete graph $K_\kk$ with $\kk$ vertices $\{v_1,v_2,\ldots,v_\kk\}$, and let $\Dom^\prime=\{1,2,\ldots,\kk\}$.
	We construct each constraint so that assigning a value $p\in \Dom$ to a vertex $v_i\in V(G^\prime)$ corresponds to choosing a vertex $\kkv{i}{p}$ as a member of a clique $Q$. 
	That is, we define $\Ecns^\prime(v_i v_j):=\{(p,q) \colon \kkv{i}{p} \kkv{j}{q}\in E(H)\}$.
	Observe that we can simultaneously assign $p$ and $q$ to $v_i$ and $v_j$, respectively, if and only if $\kkv{i}{p}$ and $\kkv{j}{q}$ are adjacent in $H$.
	Therefore, $H$ and $(G^\prime,\Dom^\prime,\Ecns^\prime)$ are equivalent.
	Clearly, $|V(G^\prime)|=|\Dom^\prime|=\kk$ holds.
	
	We now construct an instance $\Inst=(G,\Dom=\Dom^\prime \cup \{0\},\Ecns,\mapf_\ini,\mapf_\tar)$ of \CSR{2} as follows.
	A graph $G$ is obtained from $G^\prime$ by adding two new vertices $w_1$ and $w_2$ and edges $\{w_1 w \colon w\in V(G^\prime) \cup \{w_2\}\}$.
	These added vertices will form a key component which links the existence of a solution of $(G^\prime,\Dom^\prime,\Ecns^\prime)$ with the reconfigurability of $\Inst$.
	Clearly, $|V(G)|=\kk+2$ and $|\Dom|=\kk+1$ hold.
	We first construct the constraints of the original graph $G^\prime$.
	For each edge $v_i v_j\in E(G^\prime)$, we add to a constraint $\Ecns^\prime(v_i v_j)$, all mappings which contain $0$; that is, $\Ecns(v_i v_j):=\Ecns^\prime(v_i v_j)\cup (\{0\}\times \Dom) \cup (\Dom \times \{0\})$.
	We have the following observation.
	\begin{observation}
		\label{obs:ETH1}
		Solutions of $(G^\prime,\Dom,\Ecns)$ in which no vertices are assigned $0$ one-to-one corresponds to solutions of the original instance $(G^\prime,\Dom^\prime,\Ecns^\prime)$.
	\end{observation}
	We next define the constraints regarding $w_1$ and $w_2$.
	The constraint $\Ecns(w_1 w_2)$ of $w_1 w_2\in E(G)$ is defined as $\Ecns(w_1 w_2):=\{(0,1),(0,2),(1,2),(2,1)\}$.
	For each $v_i\in V(G^\prime)$, we let $\Ecns(w_1 v_i)=\Dom^2 \setminus \{(0,0)\}$.
	Then, we have the following observation.
	\begin{observation}
		\label{obs:ETH2}
		we can assign $0$ to $w_1$ if and only if no other vertices are assigned $0$.
	\end{observation}
	Finally, we define two solutions $\mapf_\ini$ and $\mapf_\tar$ as follows:
	\begin{itemize}
		\item for each vertex $v_i \in V(G^\prime)$, let $\mapf_\ini(v_i)=\mapf_\tar(v_i)=0$; and
		\item let $\mapf_\ini(w_1)=\mapf_\tar(w_2)=1$ and $\mapf_\ini(w_2)=\mapf_\tar(w_1)=2$.
	\end{itemize}
	
	In order to show the correctness, it suffices to show the following lemma.
	\begin{lemma}
		$(G^\prime,\Dom^\prime,\Ecns^\prime)$ has a solution if and only if $\Inst$ is a yes-instance.
	\end{lemma}
	We first show the if direction.
	Assume that $\mapf_\ini$ and $\mapf_\tar$ are reconfigurable.
	Then, a reconfiguration sequence must contain a solution $\mapf$ such that $\mapf(w_1)=0$.
	Otherwise, values of $w_1$ and $w_2$ can never be changed because only allowed assignment to $\{w_1, w_2\}$ is either $(1,2)$ or $(2,1)$ in this case, which contradicts the reconfigurability.
	Since $\mapf$ assigns $0$ to $w_1$, no other vertices are assigned $0$ by Observation~\ref{obs:ETH2}.	
	In addition, by Observation~\ref{obs:ETH1}, $\rest{\mapf}{V(G^\prime)}$ is a solution of $(G^\prime,\Dom^\prime,\Ecns^\prime)$.
	
	We next prove the only-if direction.
	Let $\mapg$ be a solution of $(G^\prime,\Dom^\prime,\Ecns^\prime)$.
	We extend it to solutions $\mapf_\ini^\prime$ and $\mapf_\tar^\prime$ of $(G,\Dom,\Ecns)$ as follows.
	For each $\iot \in \{\ini,\tar\}$, 
	\begin{itemize}
		\item let $\mapf_\iot^\prime(v_i)=\mapg(v_i)$ for each $v_i\in V(G^\prime)$; and
		\item let $\mapf_\iot^\prime(w_1)=\mapf_\iot(w_1)$ and $\mapf_\iot^\prime(w_2)=\mapf_\iot(w_2)$.
	\end{itemize}
	Notice that $\mapf_\iot^\prime$ is a solution, and that $\mapf_\iot$ and $\mapf_\iot^\prime$ are reconfigurable by changing values of all vertices $v_i\in V(G^\prime)$ from $0$ to $\mapf_\iot^\prime(v_i)$.
	Furthermore, $\mapf_\ini^\prime$ and $\mapf_\tar^\prime$ are reconfigurable by the following three steps:
	\begin{itemize}
		\item change a value of $w_1$ from $1$ to $0$;
		\item change a value of $w_2$ from $2$ to $1$; and
		\item change a value of $w_1$ from $0$ to $2$.
	\end{itemize}
	We note that this yields a valid reconfiguration sequence: in particular, Observation~\ref{obs:ETH2} and the construction of $\mapf_\ini^\prime$  justify the first step.
	Therefore, $\mapf_\ini$ and $\mapf_\tar$ are reconfigurable.

\section{Proof of Theorem~\ref{the:pathhard}}
	
	We give a polynomial-time reduction from $\Hwg$\textsc{-word reconfiguration}, which is defined as follows.
	Let $\Hwg$ be a (possibly non-simple) directed graph.
	An $\Hwg$\emph{-word} (of \emph{length} $\wdl$) is a string $\Hw \in V(\Hwg)^\wdl$ such that $w_i w_{i+1}\in E(\Hwg)$ for every pair of consecutive symbols $w_i$ and $w_{i+1}$ in $\Hw$; in other words, $\Hw$ can be seen as a directed walk in $\Hwg$.
	For an integer $\wdl \ge 1$, an $\Hwg$\emph{-word graph} $\Hws{\wdl}{\Hwg}$ is a graph such that $V(\Hws{\wdl}{\Hwg})$ is a set of all $\Hwg$-word of length $\wdl$, and two $\Hwg$-words $\Hw$ and $\Hw^\prime$ are adjacent if and only if the hamming distance between them is exactly one.
	Then, $\Hwg$\textsc{-word reconfiguration} asks for a given integer $\wdl$, two $\Hwg$-words $\Hw_\ini$ and $\Hw_\tar$, whether there exists a walk between $\Hw_\ini$ and $\Hw_\tar$ in $\Hws{\wdl}{\Hwg}$ or not.
	Wrochna showed that there exists a directed graph $\Hwg$ such that $\Hwg$\textsc{-word reconfiguration} is PSPACE-complete~\cite{Wro18}.
	
	\begin{figure}[t]
		\begin{center}
			\includegraphics{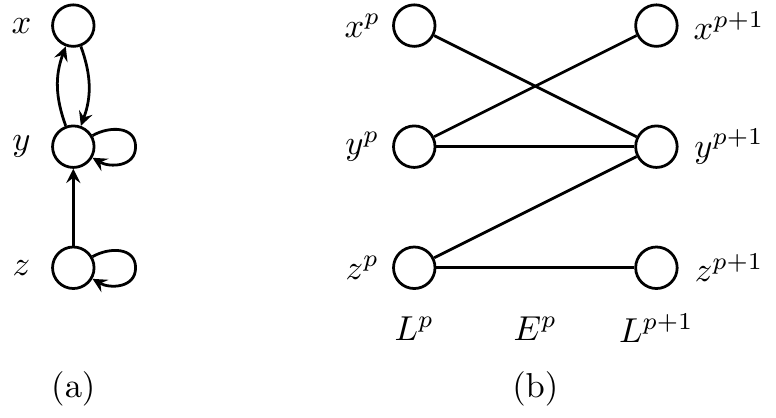}
		\end{center}
		\caption{(a) A directed graph $\Hwg$ and (b) the edge set $\Elay{p}$ between $\lay{p}$ and $\lay{p+1}$.}
		\label{fig:PWreduction}
	\end{figure}	
	
	We now construct an instance $(G,\Dom,\Ecns,\mapf_\ini,\mapf_\tar)$ of \HR{L} corresponding to an instance $(\wdl,\Hw_\ini,\Hw_\tar)$ of $\Hwg$\textsc{-word reconfiguration}.
	The idea is similar to the one used in the proof of the PSPACE-completeness of \HR{} for cycles~\cite{Wro18}.
	Let $G$ be a path with $\wdl$ vertices $v_1,v_2,\ldots,v_\wdl$.
	An \underlying{} graph $H$ is constructed as follows.
	The vertex set $V(H)$ is a union of three sets $\lay{0},\lay{1},\lay{2}$, where $\lay{p}=\{x^p \colon x\in V(\Hwg)\}$ for each $p\in \{0,1,2\}$; each $\lay{p}$ is called a \emph{layer}.
	For any integer $q>2$, we define $x^q=x^p$ and $\lay{q}=\lay{p}$ if $q\equiv p~(\mathrm{mod}~3)$.
	The edge set $E(H)$ is a union of three sets $\Elay{0},\Elay{1},\Elay{2}$, where $\Elay{p}=\{x^p y^{p+1} \colon xy\in E(\Hwg)\}$ for each $p\in \{0,1,2\}$.
	(See \figurename~\ref{fig:PWreduction} for an example of $\Elay{p}$.)
	Let $\Dom=V(H)$, and let $\Vcns(v_i)=\lay{i}$ for each vertex $v_i\in V(G)$.
	For each edge $v_i v_{i+1}\in E(G)$, we construct the constraint $\Ecns(v_i v_{i+1})$ so as to respect $H$, $\Vcns(v_i)$ and $\Vcns(v_{i+1})$.
	Finally, we define $\mapf_\iot$ for each $\iot \in \{\ini,\tar\}$ as follows.
	For each $v_i\in V(G)$, let $\mapf_\iot(v_i)=w_i^p$ if $w_i$ is the $i$-th symbol of $\Hw_\iot$.
	This completes the construction of $(G,\Dom,\Ecns,\mapf_\ini,\mapf_\tar)$, which can be done in polynomial time.
	
	In order to show the correctness, we show that there exists a bijection between all solutions of $(G,\Dom,\Ecns)$ and all $\Hwg$-words of length $\wdl$.
	Let $\mapf \colon V(G) \to \Dom$ be any mapping which respects all lists such that $(\mapf(v_1),\mapf(v_2),\allowbreak \ldots,\mapf(v_\wdl))=(w_1^1,w_2^2,\ldots,w_\wdl^\wdl)$ for some $w_1,w_2,\ldots,w_\wdl \in V(\Hwg)$.
	We now define $\omega(\mapf)$ as a string $w_1 w_2 \cdots w_\wdl$.
	For each $i \in \{1,2,\ldots,\wdl-1\}$, $\mapf(v_i) \mapf(v_{i+1})=x^i y^{i+1}\in E(H)$ if and only if $w_i w_{i+1}=xy \in E(\Hwg)$.
	Therefore, $\mapf$ is a solution of $(G,\Dom,\Ecns)$ if and only if $\omega(\mapf)$ is an $\Hwg$-word of length $\wdl$.
	Let $\omega^\prime$ be the restriction of $\phi$ on $V(\Sol{(G,\Dom,\Ecns)})$, that is, $\omega^\prime=\rest{\omega}{V(\Sol{(G,\Dom,\Ecns)})}$.
	Then, by the definition of $\omega$, $\omega^\prime$ is a bijection between $V(\Sol{(G,\Dom,\Ecns)})$ and $V(\Hws{\wdl}{\Hwg})$.
	Moreover, this bijection preserves the adjacency relation of the solution graph, and $\mapf_\iot=\omega^\prime(\Hw_\iot)$, $\iot \in \{\ini,\tar\}$.
	Thus, $(G,\Dom,\Ecns,\mapf_\ini,\mapf_\tar)$ is a yes-instance if and only if $(\wdl,\Hw_\ini,\Hw_\tar)$ is.

\section{Proofs omitted from Section~\ref{sec:reduce}}

\subsection{Proof of Lemma~\ref{lem:reduce_pre}}

	We first prove the if direction.
	Assume that $\mapf$ is a solution for $(G,\Dom,\Ecns)$.
	In order to show that $\mapf^\prime=\rest{\mapf}{V(G^\prime)}$ is a solution for $(G^\prime,\Dom,\Ecns^\prime)$, it suffices to check that $\mapf^\prime$ satisfies all constraints.
	For each $X^\prime \in E(G^\prime)$, consider the hyperedge set $E^\prime=\{X\in E(G) \colon X\setminus V_2=X^\prime \}$.
	Recall that $\Ecns^\prime(X^\prime)=\bigcap_{X\in E^\prime}\subcns{X}$, where $\subcns{X}=\{\rest{\mapg}{X^\prime} \colon \mapg \in \Ecns(X)\}$.
	For every hyperedge $X\in E^\prime$, $\rest{\mapf}{X} \in \Ecns(X)$ holds since $\mapf$ is a solution for $(G,\Dom,\Ecns)$.
	Notice that $\rest{\mapf^\prime}{X^\prime}=\rest{(\rest{\mapf}{X})}{X^\prime}$, and hence it is contained in $\subcns{X}$.
	Therefore, $\rest{\mapf^\prime}{X^\prime}\in \Ecns^\prime(X^\prime)$ holds, and hence $\mapf^\prime$ is a solution for $(G^\prime,\Dom,\Ecns^\prime)$.
	
	We next prove the only-if direction.
	Assume that $\mapf^\prime$ is a solution for $(G^\prime,\Dom,\Ecns^\prime)$.
	We claim that a mapping $\mapf \colon V(G) \to \Dom$ obtained by extending $\mapf^\prime$ as follows is a solution for $(G,\Dom,\Ecns)$:
	\[
	\mapf(v) =\left\{
	\begin{array}{ll}
	\mapf^\prime(\isov^{-1}(v)) & ~~~\mbox{if $v\in V_2$}; \\
	\mapf^\prime(v) & ~~~\mbox{otherwise}.
	\end{array} \right.
	\]
	To this end, we show that for each hyperedge $X\in E(G)$, $\rest{\mapf}{X}\in \Ecns(X)$ holds.
	Briefly, this follows the conditions (2) and (4) of Definition~\ref{def:identical}.
	If $X\in E(G^\prime)$, $\Ecns^\prime(X)$ contains $\rest{\mapf^\prime}{X}=\rest{\mapf}{X}$.
	Because $\Ecns^\prime(X)$ is a subset of $\Ecns(X)$, $\Ecns(X)$ also contains $\rest{\mapf}{X}$.
	Otherwise, $X$ is a hyperedge in $H_2^\prime=G[V_2\cup W]$.
	By the condition (4) of Definition~\ref{def:identical}, $\Ecns(X)=\repcns{\Ecns}{\widehat{\isov}}(X^\prime)$ holds, where $\widehat{\isov}=\rest{\phi}{X^\prime}$ and $\isoe(X^\prime)=X$.
	Because $W$, $V_1$ and $V_2$ are disjoint each other, $X^\prime\subseteq V_1\cup W$  does not intersect $V_2$.
	Therefore, $X^\prime$ is contained as a hyperedge in $G^\prime$, and hence $\mapf^\prime$ satisfies the constraint of $X^\prime$; that is, $\rest{\mapf^\prime}{X^\prime} \in \Ecns(X^\prime)$.
	Recall that $\repcns{\Ecns}{\widehat{\isov}}(X^\prime)=\{\mapg \circ \widehat{\isov}^{-1} \colon \mapg \in \Ecns(X^\prime)\}$.
	Thus, $(\rest{\mapf^\prime}{X^\prime})\circ \widehat{\isov}^{-1}$ is in $\Ecns(X)=\repcns{\Ecns}{\widehat{\isov}}(X^\prime)$.
	It now suffices to show that $(\rest{\mapf^\prime}{X^\prime})\circ \widehat{\isov}^{-1}=\rest{\mapf}{X}$.
	From the definition of $\mapf$, $\mapf(v)=\mapf^\prime(\isov^{-1}(v))$ holds for each $v\in V_2\cap X$, and $\mapf(v)=\mapf^\prime(v)$ holds or each $v\in W\cap X$.
	In addition, by the condition (2) of Definition~\ref{def:identical}, $\mapf(v)=\mapf^\prime(\isov^{-1}(v))$ also holds in the later case; and hence $\rest{\mapf}{X}=\rest{(\mapf^\prime \circ \isov^{-1})}{X}$.
	Since $\widehat{\isov}^{-1}=(\rest{\isov}{X^\prime})^{-1}=\rest{\isov^{-1}}{X}$ is a bijection from $X$ to $X^\prime$, we have $\rest{\mapf}{X}=\rest{(\mapf^\prime \circ \phi^{-1})}{X}=(\rest{\mapf^\prime}{X^\prime})\circ (\rest{\isov^{-1}}{X})=(\rest{\mapf^\prime}{X^\prime})\circ \widehat{\isov}^{-1}$ as required.

\subsection{Proof of Lemma~\ref{lem:reduce}}

	We first prove the if direction.
	Suppose that $\Inst$ is a yes-instance.
	Then, there exists a reconfiguration sequence $\seq{\mapf_0, \mapf_1, \ldots, \mapf_\ell}$ for $\Inst$, where $\mapf_\ell=\mapf_\tar$.
	For each $i\in \ISN{0}{\ell}$, $\rest{\mapf_i}{V(G^\prime)}$ is a solution for $(G^\prime,\Dom,\Ecns^\prime)$ by Lemma~\ref{lem:reduce_pre}.
	Therefore, by removing all duplicate solutions appearing consecutively in the resulting sequence, $\seq{\rest{\mapf_0}{V(G^\prime)}, \rest{\mapf_1}{V(G^\prime)}, \ldots,\rest{\mapf_{\ell^\prime}}{V(G^\prime)}}$ is a reconfiguration sequence for $\Inst^\prime$.
	Thus $\Inst^\prime$ is a yes-instance.
	
	We now prove the only-if direction.
	Suppose that $\Inst^\prime$ is a yes-instance.
	Then, there exists a reconfiguration sequence $\Seq^\prime=\seq{\mapf^\prime_0, \mapf^\prime_1, \ldots, \mapf^\prime_\ell}$ for $\Inst^\prime$ with $\mapf^\prime_0=\rest{\mapf_\ini}{V(G^\prime)}$ and $\mapf^\prime_\ell=\rest{\mapf_\tar}{V(G^\prime)}$.
	Our goal is to construct a reconfiguration sequence $\Seq$ for $\Inst$ from $\Seq^\prime$.
	For each $i\in \ISN{0}{\ell}$, we first extend $\mapf^\prime_i$ to $\mapf_i$ as follows:
	\[
	\mapf_i(v) =\left\{
	\begin{array}{ll}
	\mapf^\prime_i(\isov^{-1}(v)) & ~~~\mbox{if $v\in V(H_2)$}; \\
	\mapf^\prime_i(v) & ~~~\mbox{otherwise}.
	\end{array} \right.
	\]
	Notice that $\mapf_i$ corresponds to the mapping defined in the only-if proof of Lemma~\ref{lem:reduce_pre}, and hence it is a solution for $(G,\Dom,\Ecns)$.
	Therefore, $\seq{\mapf_0, \mapf_1, \ldots, \mapf_\ell}$ is a sequence of solutions for $(G,\Dom,\Ecns)$.		
	However, there may exist several indices $i\in \ISN{0}{\ell-1}$ such that $\mapf_i$ and $\mapf_{i+1}$ are not adjacent, that is, $|\diff{\mapf_i}{\mapf_{i+1}}|>1$ may hold.
	Recall that $\mapf^\prime_i$ and $\mapf^\prime_{i+1}$ are adjacent for each $i\in \ISN{0}{\ell-1}$, that is, $\diff{\mapf^\prime_i}{\mapf^\prime_{i+1}}=\{w\}$ for some vertex $w\in V(G^\prime)$.
	Therefore we know that
	\begin{itemize}
		\item if $w\notin V(H_1)$, then $\diff{\mapf_i}{\mapf_{i+1}}=\{w\}$ holds, that is, $\mapf_i$ and $\mapf_{i+1}$ are adjacent; and
		\item if $w\in V(H_1)$, then $\diff{\mapf_i}{\mapf_{i+1}}=\{w,\isov(w)\}$ holds, that is, $\mapf_i$ and $\mapf_{i+1}$ are not adjacent.
	\end{itemize}
	In the latter case, between $\mapf_i$ and $\mapf_{i+1}$, we insert a solution $\widetilde{\mapf}_i$ of $G$ defined as follows:
	\[
	\widetilde{\mapf}_i(v) =\left\{
	\begin{array}{ll}
	\mapf_{i+1}(v) & ~~~\mbox{if $v=w$}; \\
	\mapf_i(v) & ~~~\mbox{if $v=\isov(w)$}; \\
	\mapf_i(v) & ~~~\mbox{otherwise}.
	\end{array} \right.
	\]
	Observe that $\widetilde{\mapf}_i$ is a solution for $(G,\Dom,\Ecns)$.
	Moreover, both $\diff{\mapf_i}{\widetilde{\mapf}_i}=\{w\}$ and $\diff{\widetilde{\mapf}_i}{\mapf_{i+1}}=\{\isov(w)\}$ hold.
	Thus, we obtain a proper reconfiguration sequence $\Seq$ for $\Inst$ as claimed.

\subsection{Proof of Theorem~\ref{the:mweasy}}
	
	In order to prove Theorem~\ref{the:mweasy}, we first give a sufficient condition for which two subgraphs in an instance of \HR{L} satisfy Definition~\ref{def:identical}.	
	Let $\Inst=(G,\Dom,\Ecns,\mapf_\ini,\mapf_\tar)$ be a given instance of \HR{L} such that $H$ is the \underlying{} graph and $\Vcns$ is a list assignment.
	Let $H_1$ and $H_2$ are two induced subgraphs such that $|V(H_1)|=|V(H_2)|$, $V(H_1)\cap V(H_2)=\emptyset$.
	\begin{definition}
		\label{def:identicalLHR}
		$H_1$ and $H_2$ are \emph{\CR{L}-identical} if there exists a bijection $\isov \colon V(H_1) \to V(H_2)$ which satisfies the following two conditions:
		\begin{enumerate}[(i)]
			\item $vw\in E(H_1)$ if and only if $\isov(v)\isov(w)\in E(H_2)$.
			\item for every vertex $v\in V_1$,
			\begin{enumerate}[a]
				\item$\Neigh{G}{v}\setminus V(H_1)=\Neigh{G}{\isov(v)}\setminus V(H_2)$;
				\item $\asgn{v}=\asgn{\isov(v)}$, that is, $\mapf_{\ini}(v)=\mapf_{\ini}(\isov(v))$ and $\mapf_{\tar}(v)=\mapf_{\tar}(\isov(v))$; and
				\item $\Vcns(v)=\Vcns(\isov(v))$.
			\end{enumerate}
		\end{enumerate}
	\end{definition}
	\smallskip
	
	Then, we have the following lemma.
	\begin{lemma}
		\label{lem:idLHR}
		$H_1$ and $H_2$ are identical if they are \CR{L}-identical for a bijection $\isov$.
	\end{lemma}
	\begin{proof}
		We first prove that the assumption that $\Neigh{G}{V(H_1)}=\Neigh{G}{V(H_2)}$ is satisfied.
		By the condition (ii)-a, we have $\Neigh{G}{V(H_1)}=\bigcup_{v\in V(H_1)}(\Neigh{G}{v} \setminus V(H_1))=\bigcup_{v\in V(H_1)}(\Neigh{G}{\isov(v)} \setminus V(H_2))$.
		Since $\isov$ is a bijection, $\bigcup_{v\in V(H_1)}(\Neigh{G}{\isov(v)} \setminus V(H_2))=\bigcup_{v\in V(H_2)}(\Neigh{G}{v} \setminus V(H_2))=\Neigh{G}{V(H_2)}$, as required.
		
		Let $W=\Neigh{G}{V(H_1)}=\Neigh{G}{V(H_2)}$, $H_1^\prime=G[V(H_1)\cup W]$ and $H_2^\prime=G[V(H_2)\cup W]$.
		In order to prove the lemma, it suffices to give two bijections which satisfy all conditions of Definition~\ref{def:identical}.
		We define two mappings $\isov^\prime$ and $\isoe^\prime$ as follows.
		\begin{itemize}
			\item For each $v\in V(H_1^\prime)$, $\isov^\prime(v)=\isov(v)$ if $v\in V(H_1)$ and $\isov^\prime(v)=v$ if $v\in W$.
			\item For each $vw\in E(H_1^\prime)$, $\isoe^\prime(vw)=\isov^\prime(v) \isov^\prime(w)$.
		\end{itemize}
		We now prove the condition (1) for $\isov^\prime$ and $\isoe^\prime$.
		Since $\isov$ is a bijection between $V(H_1)$ and $V(H_2)$, $\isov^\prime$ is a bijection between $V(H_1^\prime)$ and $V(H_2^\prime)$.
		By the definition of $\isoe^\prime$, we suffice to show that $\isoe^\prime$ is a bijection between $E(H_1^\prime)$ and $E(H_2^\prime)$.
		From the definition and the condition (i), $\rest{\isoe^\prime}{E(G[W])}$ is a bijection (the identity mapping) between $E(G[W])$ and $E(G[W])$ and $\rest{\isoe^\prime}{E(H_1)}$ is a bijection between $E(H_1)$ and $E(H_2)$.
		Let $E_1^\prime=E(H_1^\prime)\setminus(E(H_1) \cup E(G[W]))$ and $E_2^\prime=E(H_2^\prime)\setminus(E(H_2) \cup E(G[W]))$.
		Then, it suffices to show that $\rest{\isoe^\prime}{E_1^\prime}$ is a bijection between $E_1^\prime$ and $E_2^\prime$.
		For each vertices $v\in V(H_1)$ and $w\in W$, $\isov^\prime(v) \isov^\prime(w)=\isov(v) w$ holds.
		From the condition (ii)-a, $vw\in E_1^\prime$ if and only if $\isov(v) w\in E_2^\prime$.
		Therefore, $\rest{\isoe^\prime}{E_1^\prime}$ is a bijection between $E_1^\prime$ and $E_2^\prime$, and hence $\isoe^\prime$ is a bijection between $E(H_1^\prime)$ and $E(H_2^\prime)$.
		
		The conditions (2) and (3) directly follows the definition of $\isov^\prime$ and the condition (ii)-b.
		
		We finally prove the condition (4).
		Recall that each constraint $\Ecns(vw)=\{\mapg \in \Dom^{\{v,w\}} \colon \allowbreak \mapg(v)\mapg(w) \in E(H)\cap(\Vcns(v) \times \Vcns(w))\}$ depends only on $\Vcns(v)$ and $\Vcns(w)$.
		From the condition (ii)-c, for every $vw\in E(H_1)$, $\Vcns(v) \times \Vcns(w)=\Vcns(\isov^\prime(v)) \times \Vcns(\isov^\prime(w))$ holds.
		Therefore, $\Ecns(vw)$ and $\Ecns(\isoe^\prime(vw))$ are the same when we see them as the sets of vectors.
	\end{proof}
	
	We note that Definition~\ref{def:identicalLHR} and Lemma~\ref{lem:reduce} are equivalent to the definition of ``identical subgraphs'' and a ``reduction rule'' given for \CR{L}~\cite{HIZ18}.
	Therefore, by using the same argument as \cite{HIZ18}, we have the following proposition.
	\begin{proposition}
		There exists a polynomial-time algorithm which computes from a given instance $\Inst$ of \HR{L} an instance $\Inst^\prime$ such that:
		\begin{itemize}
			\item[$\bullet$] $\Inst^\prime$ can be obtained by repeatedly applying Lemma~\ref{lem:reduce}; and
			\item[$\bullet$] the size of $\Inst^\prime$ depends only on the modular-width and the size of a maximum clique of a given graph $G$, and $\doms$.
		\end{itemize}
	\end{proposition}
	
	Because a given graph $G$ has a homomorphism to a simple graph with $\doms$ vertices, the size of a maximum clique of $G$ must be at most $\doms$.
	Thus, the size of $\Inst^\prime$ depends only on the modular-width and $\doms$, and hence we are done.

\subsection{Proof of Theorem~\ref{the:tdeasy}}
	
	We first define the tree-depth and introduce some notation.
	Let $G$ be a connected graph.
	A \emph{tree-depth decomposition} of $G$ is a rooted tree $T$ such that $V(T)=V(G)$ and if $vw \in E(G)$ then one of two endpoint is an ancestor of the other in $T$.
	The \emph{depth} of $T$ is the maximum number of vertices of a path in $T$ between the root and a leaf.
	The \emph{tree-depth} $\TD(G)$ of $G$ is the minimum depth of a tree-depth decomposition of $G$.
	It is known that there exists an algorithm computing the tree-depth decomposition of depth $\TD$ in time $O^*(2^{O(\TD^2)})$~\cite{treedepth14}.
	For a tree-depth decomposition $T$ of a graph $G$ and for a vertex $v\in V(T)=V(G)$, $\subtree{v}$ denote the subtree of $T$ rooted at $v$, and $\ances{v}$ denote the set of all ancestors of $v$.
	From the definition of a tree-depth decomposition, $\Neigh{G}{v}\subseteq V(\subtree{v})\cup \ances{v}$ holds.

\medskip
	Let $\Inst=(G,\Dom,\Ecns,\mapf_\ini,\mapf_\tar)$ be a given instance of \CSR{0}, and let $T$ be a tree-depth decomposition of $\Prm{G}$ with depth $\TD$.
	We assume that all vertices of $V(G)$ are totally ordered in the pre-order of the depth-first search on $T$ starting with its root; let $\prec$ be the binary relation defined by this ordering.
	
	We now explain a preprocessing of our kernelization algorithm, which will simplify the discussion, and then give a sufficient condition that two subhypergraphs are identical.
	
	Let $X\in E(G)$ be a hyperedge.
	From the definition of the primal graph, any two distinct vertices $v$ and $w$ in $X$ are adjacent in $\Prm{G}$, and hence they are in ancestor-descendant relationship in $T$.
	Therefore, there exists the unique vertex which is farthest from the root.
	We call $v$ the \emph{bottommost vertex} of $X$.
	Then, for each vertex $v\in V(G)$, we modify the given instance as follows. 
	\begin{itemize}
		\item Remove all hyperedges $X$ whose bottommost vertices are $v$ from $G$.
		\item Add a hyperedge $X_v=\ances{v}\cup \{v\}$ to $G$, and let $\Ecns(X_v)$ be the set of all mappings $\mapg \in \Dom^{X_v}$ which satisfy the constraints of all removed hyperedges.
	\end{itemize}
	Observe that this modification can be done in time $O^*(\doms^{\TD})$, since each vertex has at most $\TD$ ancestors, and does not change the set of solutions.
	Therefore, in the remainder of this section, we assume that there exists a one-to-one correspondence between $V(G)$ and $E(G)$ such that $v\in V(G)$ corresponds to the hyperedge $\ances{v}\cup \{v\}$.
	
	Let $v\in V(T)$ be a vertex, and assume that all vertices in $V(\subtree{v})$ are labeled as $v_1,v_2,\ldots,v_p$ according to $\prec$.
	Then, we define three $p$-tuples $\IDN{v}$, $\IDA{v}$ and $\IDC{v}$ as follows.
	\begin{itemize}
		\item The $i$-th component of $\IDN{v}$ is $\Neigh{T}{v_i}$.
		\item The $i$-th component of $\IDA{v}$ is $\asgn{v_i}$.
		\item The $i$-th component of $\IDC{v}$ is a set of vectors of length $|X_i|$ which corresponding to $\Ecns(X_i)$ under the total order $\prec$, where $X_i=\ances{v_i}\cup \{v_i\}$.
	\end{itemize}
	We call these tuples \emph{ID-tuples} of $v$.
	Then, we have the following lemma.
	\begin{lemma}
		\label{lem:IDtup}
		Let $v$ and $w$ be two children of a vertex $u$ in $T$ such that $|V(\subtree{v})|=|V(\subtree{w})|$.
		If $(\IDN{v},\IDA{v},\IDC{v})=(\IDN{w},\IDA{w},\IDC{w})$ holds, then $G[V(\subtree{v})]$ and $G[V(\subtree{w})]$ are identical under some pair of two bijections.
	\end{lemma}
	\begin{proof}
		By the preprocessing, $\Neigh{G}{V(\subtree{v})}=\Neigh{G}{V(\subtree{w})}=\ances{v}=\ances{w}$ holds.
		We denote $H_1=G[V(\subtree{v})]$, $H_2=G[V(\subtree{w})]$, $H_1^\prime=G[V(\subtree{v})\cup \ances{v}]$ and $H_2^\prime=G[V(\subtree{w})\cup \ances{w}]$.
		We now define a pair of mappings $\isov$ and $\isoe$ as follows.
		$\isov$ maps the $i$-th vertex in $V(H_1^\prime)$ to the $i$-th vertex in $V(H_2^\prime)$ according to $\prec$.
		$\isoe(X)=\{\isov(x_1),\isov(x_2),\ldots,\isov(x_\ari)\}$ holds for each hyperedge $X=\{x_1,x_2,\ldots,x_\ari\}\in E(H_1^\prime)$.
		
		Then, it suffices to show that $\isov$ and $\isoe$ satisfy the conditions (1) through (4) of Definition~\ref{def:identical}.		
		From the definition of $\isov$, the condition (2) holds; that is, $\isov(x)=x$ if $x\in \ances{v}$.
		Since $\IDA{v}=\IDA{w}$, the condition (3) holds.
		In order to verify the condition (1), we show that $\isoe$ is a bijection from $E(H_1^\prime)$ to $E(H_2^\prime)$.
		The assumption that $\IDN{v}=\IDN{w}$ implies that $\subtree{v}$ and $\subtree{w}$ isomorphic under $\isov$.
		Furthermore, $T[V(\subtree{v})\cup \ances{v}]=T[V(H_1^\prime)]$ and $T[V(\subtree{w})\cup \ances{w}]=T[V(H_2^\prime)]$ are isomorphic, too.
		Recall that for any vertex $x\in V(G)$, there exists the corresponding hyperedge $X\in E(G)$ such that $X=\ances{v}\cup \{v\}$.
		Therefore, for each hyperedge $X\in E(H_1^\prime)$ corresponding to $x\in V(H_1^\prime)$, $\isoe(X)$ is a hyperedge corresponding to $\isov(x) \in V(H_2^\prime)$.
		Thus, condition (1) holds.
		Finally, since $\IDC{v}=\IDC{w}$, the condition (3) holds.
	\end{proof}

\medskip
	We now describe our kernelization algorithm.
	Our algorithm traverses $T$ from leaves to the root, that is, the algorithm processes a vertex of $T$ after its all children are processed.
	
	Let $u\in V(G)$ be a vertex which is currently visited.
	We check if there is a pair of children $v$ and $w$ which satisfies the conditions of Lemma~\ref{lem:IDtup}.
	If such a pair is found, we apply Lemma~\ref{lem:reduce} to remove $V(\subtree{w})$.
	We note that the one-to-one correspondence between the vertex set and the hyperedge set is preserved during this process.
	Therefore, we repeat this as long as such a pair is left.
	
	The running time of this kernelization can be estimated as follows.
	For each pair of two children of a vertex, Lemma~\ref{lem:IDtup} can be checked in time polynomial in $|V(G)|$ and the maximum size of a constraint (i.e., $O(\doms^{\TD})$).
	Since there exist at most $|V(G)|^2$ pairs to be checked, the algorithm runs in time polynomial in $|V(G)|$ and $O(\doms^{\TD})$.
	
\medskip
	Finally, we prove that the obtained instance $(G^\prime,\Dom,\Ecns^\prime,\mapf^\prime_\ini,\mapf^\prime_\tar)$ has bounded size.
	We have the following lemma.
	\begin{lemma}
		\label{lem:kernel}
		The graph $G^\prime$ has at most $\tdkern{\TD}$ vertices, where $\tdkern{j}$ is recursively defined for an integer $j\ge 1$ as follows{\rm :}
		\[
		\tdkern{j} =\left\{
		\begin{array}{ll}
		1 & ~~~\mbox{if $j=1$}; \\
		\alpha^2 \cdot {(2^{\alpha} \cdot \doms^2 \cdot 2^{\doms^{\alpha}})}^{\alpha} & ~~~\mbox{otherwise},
		\end{array} \right.
		\]
		where $\alpha=\tdkern{j-1}$.
		In particular, $\tdkern{\TD}$ depends only on $\doms+\TD$.
	\end{lemma}
	\begin{proof}
		In order to prove the lemma, we define the \emph{level} of the tree-depth decomposition:
		each leaf has the level one, and the level of each internal vertex is the maximum level of a child plus one.
		Then, we prove the following stronger claim:
		\begin{claim}
			Let $u$ be a vertex of level $i$, and let $G_u$ and $T^\prime$ be the graph and the tree-depth decomposition obtained by the algorithm after processing $u$.
			Then, $G_u[V(T^\prime_u)]$ has at most $\tdkern{i}$ vertices.
		\end{claim}
		We prove the claim by the induction on $i$.		
		If $i=1$, then we have $|V(G_u[V(T^\prime_u)])|=1=\tdkern{1}$ since $u$ is a leaf.
		
		We thus assume in the remainder of the proof that $i>1$.
		Consider each child $v$ of $u$ in $T^\prime$.
		Since $v$'s level is less than $i$, and $G_u[\subtree{v}^\prime]$ is the graph obtained by the algorithm after processing $v$, $|V(\subtree{v}^\prime)|\le \tdkern{i-1}$ by the induction hypothesis.
		Thus, it suffices to show that $r$ has at most $\alpha \cdot {(2^{\alpha} \cdot \doms^2 \cdot 2^{\doms^{\alpha}})}^{\alpha}$ children, where $\alpha=\tdkern{i-1}$.
		Because the algorithm has processed $u$, there exists no pair of two children $v$ and $w$ of $u$ which satisfies Lemma~\ref{lem:IDtup}.
		Therefore, the number of children can be bounded by the number of distinct combinations of three ID-tuples with at most $\alpha$ components.
		For each $\beta \le \alpha$, the number of distinct combinations of three ID-tuples with $\beta$ components can be bounded as follows.
		Recall that for each child $v$ such that $|\subtree{v}^\prime|=\beta$, $\IDN{v}$, $\IDA{v}$, and $\IDC{v}$ consist of $\beta$ subsets of $\subtree{v}^\prime$, $\beta$ pairs of tow values from $\Dom$, and $\beta$ constraints of arity at most $\beta$, respectively.
		Thus, 
		\begin{itemize}
			\item the number of distinct $\IDN{\cdot}$'s is at most ${(2^{\beta})}^\beta$;
			\item the number of distinct $\IDA{\cdot}$'s is at most ${(\doms^2)}^\beta$; and
			\item the number of distinct $\IDC{\cdot}$'s is at most ${(2^{\doms^\beta})}^\beta$.
		\end{itemize}
		Therefore, the number of children is at most
		\[
		\sum_{\beta=1}^\alpha {(2^{\beta} \cdot \doms^2 \cdot 2^{\doms^{\beta}})}^{\beta} \le \alpha \cdot {(2^{\alpha} \cdot \doms^2 \cdot 2^{\doms^{\alpha}})}^{\alpha}
		\]
		as required.
	\end{proof}

	This completes the proof of the theorem.

\section{Proof of Theorem~\ref{the:vc}}

	We give a fixed-parameter algorithm for a more general ``weighted variant'' of \CSR{0}.
	Let $\Inst=(G,\Dom,\Ecns,\mapf_\ini,\mapf_\tar)$ be an instance of \CSR{0}, and assume that each vertex $v\in V(G)$ has a \emph{weight} $\weight(v)\in \NN$, where $\NN$ is the set of all positive integers.
	According to the weight function $\weight$, we define the weight of each edge $\mapf \mapf^\prime \in E(\Sol{(G,\Dom,\Ecns)})$ in the solution graph as the weight of $v$, where $\{v\}=\diff{\mapf}{\mapf^\prime}$.
	We denote by $\Solw{\weight}{(G,\Dom,\Ecns)}$ the weighted solution graph defined in this way.
	The \emph{length} $\len{\weight}{\Seq}$ of a reconfiguration sequence (i.e., a walk) $\Seq$ in $\Solw{\weight}{(G,\Dom,\Ecns)}$ is the sum of the weight of the edges in $\Seq$.
	For each vertex $v\in V(G)$, we denote by $\nrec{\Seq}{v}$ the number of edges $\mapf \mapf^\prime$ in $\Seq$ such that $\diff{\mapf}{\mapf^\prime}=\{v\}$.
	In other words, $\nrec{\Seq}{v}$ is the number of steps changing the value of $v$ in $\Seq$.
	Notice that $\len{\weight}{\Seq}=\sum_{v\in V(G)}\weight(v)\cdot \nrec{\Seq}{v}$ holds.
	We denote by $\OPT{\Inst}{\weight}$ the minimum length of a reconfiguration sequence between $\mapf_\ini$ and $\mapf_\tar$; we define $\OPT{\Inst}{\weight}=+\infty$ if $\Inst$ is a no-instance of \CSR{0}.
	The \emph{weighted variant} of \CSR{0} is to determine whether $\OPT{\Inst}{\weight}\le \ell$ or not for a given instance $(\Inst,\weight)$ and an integer $\ell \ge 0$.
	Notice that the shortest variant is equivalent to the weighted variant where every vertex has weight one.
	Thus, in order to prove Theorem~\ref{the:vc}, it suffices to construct a fixed-parameter algorithm for the weighted variant parameterized by $\doms+\VC$.
	
	The central idea is the same as Section~\ref{sec:kernelFPT}; that is, we again kernelize the input instance.
	
\subsection{Reduction rule for the weighted variant}

	In this subsection, we give the counterpart of Lemma~\ref{lem:reduce} for the weighted version.
	
	Let $(\Inst=(G,\Dom,\Ecns,\mapf_{\ini},\mapf_{\tar}),\weight)$ be an instance of the weighted version, and assume that there exist two identical subgraphs $H_1$ and $H_2$ of $G$, both of which consist of single vertices, say, $V(H_1)=\{v_1\}$ and $V(H_2)=\{v_2\}$.
	We now define a new instance $(\Inst^\prime,\weight^\prime)$ as follows:
	\begin{itemize}
		\item $\Inst^\prime$ is the instance obtained by applying Lemma~\ref{lem:reduce} for $H_1$ and $H_2$; and
		\item $\weight^\prime(v_1)=\weight(v_1)+\weight(v_2)$ and $\weight^\prime(v)=\weight(v)$ for any $v\in V(G)\setminus \{v_1,v_2\}$.
	\end{itemize}
	Intuitively, $v_2$ is merged into $v_1$ together with its weight.
	Then, we have the following lemma.
	\begin{lemma}
		\label{lem:wreduce}
		$\OPT{\Inst}{\weight}=\OPT{\Inst^\prime}{\weight^\prime}$.
	\end{lemma}
	\begin{proof}
		Let $\Inst^\prime=(G^\prime,\Dom,\mapf^\prime_\ini,\mapf^\prime_\tar)$.
		By Lemma~\ref{lem:reduce}, $\OPT{\Inst}{\weight}=+\infty$ if and only if $\OPT{\Inst^\prime}{\weight^\prime}=+\infty$.
		Therefore, we assume that $\OPT{\Inst^\prime}{\weight^\prime}\neq +\infty$ and $\OPT{\Inst}{\weight}\neq +\infty$.	
		
		We first show that $\OPT{\Inst}{\weight}\le \OPT{\Inst^\prime}{\weight^\prime}$.
		Since $\OPT{\Inst}{\weight}\le \len{\weight}{\Seq}$ holds for any reconfiguration sequence $\Seq$ for $\Inst$, it suffices to show that there exists a reconfiguration sequence for $\Inst$ whose length is at most $\OPT{\Inst^\prime}{\weight^\prime}$.
		Let $\Seq^\prime$ be a shortest reconfiguration sequence for $\Inst^\prime$ such that $\len{\weight^\prime}{\Seq^\prime}=\OPT{\Inst^\prime}{\weight^\prime}$.
		Following the only-if direction proof of Lemma~\ref{lem:reduce}, we can construct a reconfiguration sequence $\Seq$ for $\Inst$ such that $\nrec{\Seq}{v_1}=\nrec{\Seq}{v_2}=\nrec{\Seq^\prime}{v_1}$ and $\nrec{\Seq}{v}=\nrec{\Seq^\prime}{v}$ for any $v\in V(G)\setminus \{v_1,v_2\}$.
		Therefore, 
		\[
		\begin{array}{lll}
		\len{\weight}{\Seq} &=& \sum_{v\in V(G)}\weight(v)\cdot \nrec{\Seq}{v} \\
		&=& \weight(v_1) \cdot \nrec{\Seq}{v_1} + \weight(v_2) \cdot \nrec{\Seq}{v_2} + \sum_{v\in V(G)\setminus \{v_1,v_2\}}\weight(v)\cdot \nrec{\Seq}{v} \\
		&=& (\weight(v_1) + \weight(v_2))\cdot \nrec{\Seq}{v_1} + \sum_{v\in V(G)\setminus \{v_1,v_2\}}\weight(v)\cdot \nrec{\Seq}{v} \\
		&=& \weight^\prime(v_1) \cdot \nrec{\Seq^\prime}{v_1} + \sum_{v\in V(G)\setminus \{v_1,v_2\}}\weight^\prime(v)\cdot \nrec{\Seq^\prime}{v} \\
		&=& \sum_{v\in V(G^\prime)}\weight^\prime(v)\cdot \nrec{\Seq^\prime}{v} \\
		&=& \len{\weight^\prime}{\Seq^\prime} \\
		&=& \OPT{\Inst^\prime}{\weight^\prime}.
		\end{array}
		\]
		Thus, $\Seq$ is a desired reconfiguration sequence for $\Inst$.
		
		We next show that $\OPT{\Inst^\prime}{\weight^\prime}\le \OPT{\Inst}{\weight}$.
		Since $\OPT{\Inst^\prime}{\weight^\prime}\le \len{\weight^\prime}{\Seq^\prime}$ holds for any reconfiguration sequence $\Seq^\prime$ for $\Inst^\prime$, it suffices to show that there exists a reconfiguration sequence for $\Inst^\prime$ whose length is at most $\OPT{\Inst}{\weight}$.
		Let $\Seq$ be a shortest reconfiguration sequence for $\Inst$ such that $\len{\weight}{\Seq}=\OPT{\Inst}{\weight}$.
		We now construct a reconfiguration sequence for $\Inst^\prime$ from $\Seq$ such that $\len{\weight^\prime}{\Seq^\prime}\le \OPT{\Inst}{\weight}$ as follows.
		
		\medskip
		\noindent 
		\textbf{Case 1.} $\nrec{\Seq}{v_1}\le \nrec{\Seq}{v_2}$.\\
		\noindent
		In this case, we restrict all solutions in $\Seq$ on $V(G^\prime)$ to obtain a reconfiguration sequence $\Seq_1$ for $\Inst^\prime$; recall the if direction proof of Lemma~\ref{lem:reduce}.
		From the construction, $\nrec{\Seq_1}{v_1}=\nrec{\Seq}{v_1}\le \nrec{\Seq}{v_2}$ and $\nrec{\Seq_1}{v}=\nrec{\Seq}{v}$ holds for any vertex $v\in V(G^\prime)=V(G) \setminus \{v_2\}$.
		Therefore, we have
		\[
		\begin{array}{lll}
		\len{\weight^\prime}{\Seq_1} &=& \sum_{v\in V(G^\prime)}\weight^\prime(v)\cdot \nrec{\Seq_1}{v} \\
		&=& \weight^\prime(v_1) \cdot \nrec{\Seq_1}{v_1} + \sum_{v\in V(G^\prime)\setminus \{v_1\}}\weight^\prime(v)\cdot \nrec{\Seq_1}{v} \\
		&=& (\weight(v_1)+\weight(v_2)) \cdot \nrec{\Seq}{v_1} + \sum_{v\in V(G)\setminus \{v_1,v_2\}}\weight(v)\cdot \nrec{\Seq}{v} \\
		&\le & \weight(v_1) \cdot \nrec{\Seq}{v_1} + \weight(v_2) \cdot \nrec{\Seq}{v_2} + \sum_{v\in V(G)\setminus \{v_1,v_2\}}\weight(v)\cdot \nrec{\Seq}{v} \\
		&=& \sum_{v\in V(G)}\weight(v)\cdot \nrec{\Seq}{v} \\
		&=& \len{\weight}{\Seq} \\
		&=& \OPT{\Inst}{\weight}.
		\end{array}
		\]
		Thus, $\Seq_1$ is a desired reconfiguration sequence for $\Inst^\prime$.
		
		\medskip
		\noindent 
		\textbf{Case 2.} $\nrec{\Seq}{v_1} > \nrec{\Seq}{v_2}$.\\
		\noindent				
		In this case, instead of restricting solutions in $\Seq$ on $V(G^\prime)=V(G)\setminus \{v_2\}$, we restrict them on $V(G)\setminus \{v_1\}$ and obtain a reconfiguration sequence $\Seq_2$ for an instance obtained by restricting on $V(G)\setminus \{v_1\}$.
		Then, because $H_1$ and $H_2$ are identical, we can easily ``rephrase'' $\Seq_2$ as a reconfiguration sequence $\Seq_2^\prime$ for $\Inst^\prime$.
		By the same arguments as the case 1 above, we have $\len{\weight^\prime}{\Seq_2^\prime} < \len{\weight}{\Seq}= \OPT{\Inst}{\weight}$.
		Thus, $\Seq_2^\prime$ is a desired reconfiguration sequence for $\Inst^\prime$.
		
		In this way, we have shown that $\OPT{\Inst}{\weight}=\OPT{\Inst^\prime}{\weight^\prime}$ as claimed.
	\end{proof}
	
\subsection{Kernelization}

	Finally, we give a kernelization algorithm as follows.
	
	Let $(\Inst=(G,\Dom,\Ecns,\mapf_{\ini},\mapf_{\tar}),\weight)$ be an instance of the weighted version such that the primal graph $\Prm{G}$ has a vertex cover of size at most $\VC$.
	Because such a vertex cover can be computed in time $O(2^\VC \cdot n)$~\cite{DF99}, assume that we are given a vertex cover $\Cov$ of size at most $\VC$.
	Notice that $\Vind:=V(G)\setminus \Cov$ forms an independent set of $\Prm{G}$.
	In order to simplify the proof, we first modify $\Inst$ without changing the set of solutions as follows.
	For each vertex $v$ in $\Vind$, we remove all hyperedges containing $v$ and add a new hyperedge $\Cov_v=\Cov \cup \{v\}$; note that each removed hyperedges are subsets of $\Cov_v$ since $\Vind$ is an independent set of $\Prm{G}$.
	We then define a constraint of $\Cov_v$ as the set of all mappings $\mapg \in \Dom^{\Cov_v}$ satisfying the constraints of all removed hyperedges.
	Observe that this modification can be done in time $O^*(\doms^{O(\VC)})$ and does not change the set of solutions.
	Thus, we can assume that for each vertex $v\in \Vind$, there exists the unique hyperedge $\Cov_v$ such that $\Cov_v=\Cov \cup \{v\}$.
	
	Let $v_1$ and $v_2$ be two vertices in $\Vind$.
	We now define two mappings $\isov \colon \Cov_{v_1} \to \Cov_{v_2}$ and $\isoe \colon \{\Cov_{v_1}\} \to \{\Cov_{v_2}\}$ as follows:
	$\isov(w)=w$ if $w\in \Cov$, $\isov(v_1)=v_2$, and $\isoe(\Cov_{v_1})=\Cov_{v_2}$.
	Suppose that $\asgn{v_1}=\asgn{v_2}$ and $\Ecns(\Cov_{v_2})=\repcns{\Ecns}{\isov}(\Cov_{v_1})$ hold.
	Then, induced subgraphs $G[\{v_1\}]$ and $G[\{v_2\}]$ are identical.
	Therefore, we can apply Lemma~\ref{lem:wreduce} to remove $v_2$ from $G$, and modify the weight function without changing the optimality.
	As a kernelization, we repeatedly apply Lemma~\ref{lem:wreduce} for all such pairs of vertices in $\Vind$, which can be done in polynomial time.
	Let $G^\prime$ be the resulting subgraph of $G$, and let $\Vind^\prime:=V(G^\prime)\setminus \Cov$.
	Since $\Cov$ is of size at most $\VC$, it suffices to prove the following lemma.
	\begin{lemma}
		\label{lem:kernel_sp}
		$|\Vind^\prime|\le k^2 \cdot 2^{k^{\VC+1}}$.
	\end{lemma}
	\begin{proof}
		Recall that $\Vind^\prime$ contains no pair of vertices which correspond to identical subgraphs, and hence any pair of vertices $v_1,v_2\in \Vind^\prime$ does not satisfy at least one of $\asgn{v_1}=\asgn{v_2}$ and $\Ecns(\Cov_{v_2})=\repcns{\Ecns}{\isov}(\Cov_{v_1})$.
		Therefore, $|\Vind^\prime|$ can be bounded by the number of distinct combinations of a vertex assignment and a constraint of arity $\VC+1$.
		Since the domain has size $\doms$, the number of (possible) vertex assignments can be bounded by $k^2$.
		Since a constraint of arity $\VC+1$ can be seen as a subset of $\Dom^{\VC+1}$, the number of (possible) constraints can be bounded by $2^{k^{\VC+1}}$.
		We thus have $|\Vind^\prime|\le k^2 \cdot 2^{k^{\VC+1}}$ as claimed.
	\end{proof}
	
	This completes the proof of Theorem~\ref{the:vc}.	

\section{Proof of Theorem~\ref{the:vcfast}}
	
	In order to prove the theorem, we first introduce the notion of a ``contracted solution graph'', which was first introduced in \cite{Bon17} and used in several literatures such as \cite{BP16,HIZ14}.
	
	\medskip
	Let $\Inst=(\Instcs,\mapf_\ini,\mapf_\tar)$ be an instance of \CSR{0}, where $\Instcs=(G,\Dom,\Ecns)$, and let $\Part$ be a partition of the vertex set of the solution graph $\Sol{\Instcs}$.
	The \emph{contracted solution graph} (or \emph{CSG} for short) $\CSG{\Instcs}{\Part}$ is defined as follows.
	The vertex set $V(\CSG{\Instcs}{\Part})$ is exactly $\Part$; we call each vertex of the CSG a \emph{node}.
	Each pair of distinct nodes (i.e., sets of solutions) $\Cnd,\Cnd^\prime\in \Part$ are adjacent in the CSG if and only if there exist two solutions $\mapf \in \Cnd$ and $\mapf^\prime \in \Cnd^\prime$ such that $\mapf \mapf^\prime \in E(\Sol{\Instcs})$.
	In other words, $\CSG{\Instcs}{\Part}$ is obtained by contracting a (possibly disconnected) subgraph of $\Sol{\Instcs}$ induced by each set $\Cnd \in \Part$ into one node.
	A partition $\Part$ is \emph{proper} if every set $\Cnd \in \Part$ induces a connected subgraph of $\Sol{\Instcs}$.
	Since the contraction of a connected subgraph maintains the connectivity of a graph, we have the following proposition.
	\begin{proposition}
		Let $\Inst=(\Instcs,\mapf_\ini,\mapf_\tar)$ be an instance of \CSR{0}, where $\Instcs=(G,\Dom,\Ecns)$, and let $\Part$ be a proper partition of $V(\Sol{\Instcs})$.
		Then, $\Inst$ is a yes-instance if and only if there exists a walk between $\Cnd_\ini$ and $\Cnd_\tar$ in $\CSG{\Instcs}{\Part}$, where $\mapf_\ini \in \Cnd_\ini$ and $\mapf_\tar \in \Cnd_\tar$.
		Moreover, the above condition can be checked in time polynomial in $|\Part|$.
	\end{proposition}
	
	Therefore, we first define a proper partition $\Part$ such that $|\Part|$ depends only on $\doms+\VC$, and then give an algorithm constructing the CSG and specifying the nodes corresponding to $\mapf_\ini$ and $\mapf_\tar$.
	
\subsection{Defining a proper partition}
	\label{sec:partition}
	
	Let $\Inst=(\Instcs,\mapf_\ini,\mapf_\tar)$ be an instance of \CSR{0}, where $\Instcs=(G,\Dom,\Ecns)$.
	Assume that $\Prm{G}$ has a vertex cover $\Cov$ of size at most $\VC$.
	For each solution $\mapf \in V(\Sol{\Instcs})$, we define $\eqcls{\mapf}=\{\mapf^\prime \colon \rest{\mapf}{\Cov}=\rest{\mapf^\prime}{\Cov}\}$.
	Then, we define $\Part=\{\eqcls{\mapf} \colon \mapf \in V(\Sol{\Instcs})\}$; that is, $\Part$ is the set of the equivalence classes under the equivalence relation ``their restrictions on $\Cov$ are the same''.
	Clearly, $\Part$ is a partition of $V(\Sol{\Instcs})$ and $|\Part|$ is bounded by the number of mappings from $\Cov$ to $\Dom$, that is, $|\Part|\le \doms^\VC$.
	
	In order to prove that $\Part$ is proper, we introduce some notation.
	Let $S\subseteq V(G)$ be a vertex subset, and let $\maph \colon S \to \Dom$ be a mapping from $S$ to $\Dom$.
	We define the \emph{substitution} $\sbst{\Instcs}{\maph}$ as an instance $(G^\prime,\Dom,\Ecns^\prime)$ of \CSP{0} such that:
	\begin{itemize}
		\item $G^\prime=G\setminus S$; and
		\item for each $X^\prime \in E(G^\prime)$, $\Ecns^\prime(X^\prime)=\bigcap_{X\in E^\prime}\subcns{X}$, 
		where $E^\prime=\{X\in E(G) \colon X\setminus S=X^\prime \}$ and $\subcns{X}=\{\rest{\mapg}{X^\prime} \colon \mapg \in \Ecns(X),\allowbreak \mbox{$\maph$ and $\mapg$ are compatible}\}$.
	\end{itemize}
	We have the following lemma.
	\begin{lemma}
		\label{lem:substitution}	
		Let $\mapf^\prime \colon V(G)\setminus S \to \Dom$ and $\mapf \colon V(G) \to \Dom$ be two mappings such that $\rest{\mapf}{V(G)\setminus S}=\mapf^\prime$.
		Then, $\mapf^\prime$ is a solution for $\sbst{\Instcs}{\rest{\mapf}{S}}=(G^\prime,\Dom,\Ecns^\prime)$ if and only if $\mapf$ is a solution for $(G,\Dom,\Ecns)$.
	\end{lemma}
	\begin{proof}		
		We first show the if direction.
		Assume that $\mapf$ is a solution for $(G,\Dom,\Ecns)$.
		Let $X^\prime$ be any hyperedge of $G^\prime$.
		For every hyperedge $X$ of $G$ such that $X\setminus S=X^\prime$, $\rest{\mapf}{X} \in \Ecns(X)$ holds.
		Since $\rest{\mapf}{S}$ and $\rest{\mapf}{X}$ are compatible, by the definition of $\Ecns^\prime$, $\rest{(\rest{\mapf}{X})}{X^\prime}$ is in $\Ecns^\prime(X^\prime)$.
		In addition, $\rest{(\rest{\mapf}{X})}{X^\prime}=\rest{\mapf}{X^\prime}=\rest{\mapf^\prime}{X^\prime}$ holds, and hence, $\mapf^\prime$ satisfies $\Ecns^\prime(X^\prime)$.
		Therefore, $\mapf^\prime$ is a solution for $\sbst{\Instcs}{\rest{\mapf}{S}}$.
		
		We next show the only-if direction.
		Assume that $\mapf^\prime$ is a solution for $\sbst{\Instcs}{\rest{\mapf}{S}}=(G^\prime,\Dom,\Ecns^\prime)$.
		Let $X$ be any hyperedge of $G$.
		Then, there exists a hyperedge $X^\prime$ of $G^\prime$ such that $X\setminus S=X^\prime$.
		From the definition of $\Ecns^\prime$, there exists a mapping $\mapg \in \Ecns(X)$ such that $\rest{\mapf^\prime}{X^\prime}=\rest{\mapg}{X^\prime}$, and $\rest{\mapf}{S}$ and $\mapg$ are compatible, that is, $\rest{\mapf}{S\cap X}=\rest{\mapg}{S\cap X}$.
		Since $\rest{\mapf^\prime}{X^\prime}=\rest{\mapf}{X}$ and $X^\prime \cup (S\cap X)=X$, $\rest{\mapf}{X}=\rest{\mapg}{X}$ holds, and hence $\mapf$ satisfies $\Ecns(X)$.
		Therefore, $\mapf$ is a solution for $(G,\Dom,\Ecns)$.
	\end{proof}
	
	The following lemma implies that $\Part$ is proper.
	\begin{lemma}
		\label{lem:vcproper}
		Let $\Cnd$ be a solution set in $\Part$ such that $\rest{\mapf}{\Cov}=\maph$ holds for every $\mapf \in \Cnd$.
		Then, $\Sol{\Instcs}[\Cnd]$ is connected.
	\end{lemma}
	\begin{proof}
		By Lemma~\ref{lem:substitution}, there exists a one-to-one correspondence between $\Cnd$ and the solution set for $\sbst{\Instcs}{\maph}$ which preserves the adjacency relation.
		Thus, we suffices to show that $\Sol{\sbst{\Instcs}{\maph}}$ is connected.
		Since $\Cov$ is a vertex cover of $\Prm{G}$, $\Prm{G}[V(G)\setminus \Cov]$ has no edges.
		This means that $\sbst{\Instcs}{\maph}$ contains only $1$-ary constraints.
		Therefore, a value assignment of each vertex $v\in V(G)\setminus \Cov$ can be changed independently, and hence $\Sol{\sbst{\Instcs}{\maph}}$ is connected.
	\end{proof}
	
\subsection{Algorihm computing CSG}
	\label{sec:vcalg}
	
	In order to give an algorithm computing $\CSG{\Instcs}{\Part}$ correctly, we first show two claims.
	\begin{claim}
		\label{clm:vcnode}
		Let $\maph$ be a mapping from $\Cov$ to $\Dom$.
		Then, $\CSG{\Instcs}{\Part}$ has a node corresponding to $\maph$ if and only if $\sbst{\Instcs}{\maph}=(G^\prime,\Dom,\Ecns^\prime)$ has a solution.
	\end{claim}
	\begin{proof}
		By Lemma~\ref{lem:substitution}, $\sbst{\Instcs}{\maph}$ has a solution $\mapf^\prime$ if and only if there exists a solution $\mapf$ for $\Instcs$ such that $\rest{\mapf}{V(G^\prime)}=\mapf^\prime$ and $\rest{\mapf}{\Cov}=\maph$.
		Since $\Part$ is a partition of the solution set, there exists a set $\Cnd \in \Part$ which contains $\mapf$; and hence $\CSG{\Instcs}{\Part}$ has a node corresponding to $\maph$.
	\end{proof}
	\begin{claim}
		\label{clm:vcedge}
		Let $\Cnd_1$ and $\Cnd_2$ be two nodes of $\CSG{\Instcs}{\Part}$, and let $\maph_1 \colon \Cov \to \Dom$ and $\maph_2 \colon \Cov \to \Dom$ be mappings corresponding to $\Cnd_1$ and $\Cnd_2$, respectively.
		Then, $\Cnd_1 \Cnd_2 \in E(\CSG{\Instcs}{\Part})$ if and only if both of the following conditions hold:
		\begin{itemize}
			\item []
			\begin{itemize}
				\item[$\bullet$] $|\diff{\maph_1}{\maph_2}|=1$; and
				\item[$\bullet$] $\sbst{\Instcs}{\maph_1}$ and $\sbst{\Instcs}{\maph_2}$ has a common solution $\mapf^\prime$.
			\end{itemize}
		\end{itemize}
	\end{claim}
	\begin{proof}
		We first assume that the above two conditions hold.
		For each $i\in \{1,2\}$, let $\mapf_i \colon V(G) \to \Cov$ be a mapping such that $\rest{\mapf_i}{V(G^\prime)}=\mapf^\prime$ and $\rest{\mapf_i}{\Cov}=\maph_i$.
		Then, by Lemma~\ref{lem:substitution}, $\mapf_1$ and $\mapf_2$ are adjacent solutions such that $\mapf_1 \in \Cnd_1$ and $\mapf_2 \in \Cnd_2$; and hence $\Cnd_1 \Cnd_2 \in E(\CSG{\Instcs}{\Part})$.
		
		We next assume that $\Cnd_1 \Cnd_2 \in E(\CSG{\Instcs}{\Part})$, that is, there exist two solutions $\mapf_1$ and $\mapf_2$ such that $|\diff{\mapf_1}{\mapf_2}|=1$, $\mapf_1 \in \Cnd_1$ and $\mapf_2 \in \Cnd_2$.
		From the definition of $\Part$, $\maph_1\ne \maph_2$.
		Therefore, $\diff{\mapf_1}{\mapf_2}=\diff{\maph_1}{\maph_2}$, and hence the first condition holds.
		Furthermore, $\rest{\mapf_1}{V(G)\setminus \Cov}=\rest{\mapf_2}{V(G)\setminus \Cov}$ holds.
		By Lemma~\ref{lem:substitution}, $\rest{\mapf_1}{V(G)\setminus \Cov}$ and $\rest{\mapf_2}{V(G)\setminus \Cov}$ are solutions for $\sbst{\Instcs}{\maph_1}$ and $\sbst{\Instcs}{\maph_2}$, respectively; and hence the second condition hold.
	\end{proof}
	
	From Claims~\ref{clm:vcnode} and \ref{clm:vcedge}, we can construct the following algorithm to compute $\CSG{\Instcs}{\Part}$ with nodes corresponding to $\mapf_\ini$ and $\mapf_\tar$.
	\begin{description}
		\item[Phase 1]
		For each mapping $\maph$ from $\Cov$ to $\Dom$, check if $\sbst{\Instcs}{\maph}$ has a solution.
		If so, create a node corresponding to $\maph$.
		For each $\iot \in \{\ini,\tar\}$, if $\maph=\rest{\mapf_\iot}{\Cov}$, it corresponds to $\mapf_\iot$.
		\item[Phase 2]
		For each pair of two nodes $\Cnd_1$ and $\Cnd_2$, check if the two conditions of Claim~\ref{clm:vcedge} hold.
		If so, join them by an edge.
	\end{description}
	
	The correctness follows from Claims~\ref{clm:vcnode} and \ref{clm:vcedge}.
	The first phase can be done in polynomial time for each mapping, because the constructed instance $\sbst{\Instcs}{\maph}$ of \CSP{0} contains only $1$-ary constraints.
	Since $|\Dom^\Cov|\le \doms^\VC$, whole running time of this phase is $O^*(\doms^\VC)$.
	In the second phase, the second condition of Claim~\ref{clm:vcedge} can be checked as follows.
	Let $\Ecns_1$ and $\Ecns_2$ are constraint assignments in the substitutions $\sbst{\Instcs}{\maph_1}$ and $\sbst{\Instcs}{\maph_2}$.
	We now define for each $X^\prime \in E(G^\prime)$ a constraint $\Ecns^\prime(X^\prime)=\Ecns_1(X^\prime)\cap \Ecns_2(X^\prime)$.
	Then, a solution for $(G^\prime,\Dom,\Ecns^\prime)$ is also a solution for both of $\sbst{\Instcs}{\maph_1}$ and $\sbst{\Instcs}{\maph_2}$.
	Because $(G^\prime,\Dom,\Ecns^\prime)$ is an instance of \CSP{0} which contains only $1$-ary constraints, we can solve it in polynomial time.
	Therefore, whole running time of this phase is $O^*(\doms^{O(\VC)})$.
	
	We thus completed the proof of Theorem~\ref{the:vcfast}.

\section{Proof of Theorem~\ref{the:nb}}

	Let $\Inst=(\Instcs,\mapf_\ini,\mapf_\tar)$ be an instance of \CSR{2}, where $\Instcs=(G,\Dom,\Ecns)$.
	We denote by $\Bs$ and $\NBs$ be the set of Boolean and non-Boolean vertices, respectively.
	We first define a partition similarly to Section~\ref{sec:partition} as follows.
	For each solution $\mapf \in V(\Sol{\Instcs})$, we define $\eqcls{\mapf}=\{\mapf^\prime \colon \rest{\mapf}{\NBs}=\rest{\mapf^\prime}{\NBs}\}$.
	Then, we define $\Part=\{\eqcls{\mapf} \colon \mapf \in V(\Sol{\Instcs})\}$.
	Clearly, $\Part$ is a partition of $V(\Sol{\Instcs})$ and $|\Part|\le \doms^\NB$.
	In contrast to Section~\ref{sec:partition}, however, $\Part$ may be improper in this case.\footnote{
		As a simple example, let us consider $2$-colorings of $K_2$.
		Clearly $\NBs$ is empty set, and hence all (indeed, only two) $2$-colorings of $G$ are in the same set of $\Part$.
		On the other hand, they are not reconfigurable each other.}
	
	Therefore, as the first step of our algorithm, we modify the given instance so that $\Part$ is proper by some preprocessing.
	More formally, we show the following lemma.
	\begin{lemma}
		\label{lem:nbcore}
		Let $\Inst=(\Instcs,\mapf_\ini,\mapf_\tar)$ be an instance of \CSR{2}, where $\Instcs=(G,\Dom,\Ecns)$.
		We can compute in polynomial time an instance $\Inst^\ppsup=(\Instcs^\ppsup,\mapf_\ini^\ppsup,\mapf_\tar^\ppsup)$ of \CSR{2} such that:
		\begin{enumerate}
			\item the number of non-Boolean vertices in $\Instcs^\ppsup$ is at most that of $\Instcs$;
			\item $\Inst$ is a yes-instance if and only if $\Inst^\ppsup$ is; and
			\item the partition for $\Instcs^\ppsup$ is proper.
		\end{enumerate}
	\end{lemma}
	
	Then, we can compute the CSG in the same way as Section~\ref{sec:vcalg}; we just replace $\Cov$ with $\NBs$.
	In order to check the existence of vertices and edges in the CSG, we solve instances of \textsc{Boolean} \CSP{2} which are constructed by the substitution operation.
	Since \textsc{Boolean} \CSP{2} can be solved in polynomial time~\cite{Schaefer}, the whole running time of the algorithm is $O^*(\doms^{O(\NB)})$.
	In the remainder of this section, we prove Lemma~\ref{lem:nbcore}.
	
\subsection{Implication graphs}
	
	In order to describe a preprocessing, we first introduce the notion of ``implication graph'', which was first introduced in~\cite{GKMP09} in order to prove the tractability of some variant of \textsc{Boolean} \CSR{2}.
	Let $\Instcs=(G,\Dom,\Ecns)$ be an instance of \CSP{0} where $\{0,1\}$ is a domain; we consider the values $0$ and $1$ as Boolean values.
	We define the \emph{implication graph} $\Imp{\Instcs}$ for $\Instcs$ as follows.
	For each vertex $v\in V(G)$ and for each value $i\in \Dom$ such that there exists a solution $\mapf$ with $\mapf(v)=i$, we add a vertex $\vimp{v}{i}$ to $\Imp{\Instcs}$.
	For each adjacent vertices $v,w\in V(G)$, add two arcs $\vimp{v}{i} \to \vimp{w}{\neg{j}}$ and $\vimp{w}{j} \to \vimp{v}{\neg{i}}$ if and only if $(i,j)\notin \Ecns(vw)$, where $\neg{\ }$ denote a negation of a Boolean value.
	Intuitively, an arc $\vimp{v}{i} \to \vimp{w}{\neg{j}}$ means that if $v$ is assigned $i$, then $w$ must be assigned $\neg{j}$ in any solution.
	We note that $\Imp{\Instcs}$ can be computed in polynomial time, since the existence of a vertex $\vimp{v}{i}$ can be checked in polynomial time by solving \textsc{Boolean} \CSP{2} instance obtained by substituting $\mapf(v)=i$.
	We now prove the following lemma.
	\begin{lemma}
		\label{lem:imp}
		If there exists a vertex $v\in V(G)$ such that $\vimp{v}{0}$ or $\vimp{v}{1}$ is contained in a directed cycle of $\Imp{\Instcs}$, any two solutions $\mapf_0$ and $\mapf_1$ for $\Instcs$ such that $\mapf_0(v)=0$ and $\mapf_1(v)=1$ are not reconfigurable.
		On the other hand, if $\Imp{\Instcs}$ contains no directed cycles, $\Sol{\Instcs}$ is connected.
	\end{lemma}
	\begin{proof}
		The second statement can be proved by the similar argument as the proof of Lemma~4.9 in~\cite{GKMP09}, although our implication graph is slightly different from the original one.

		Therefore, we prove the first statement.
		Let $C=\vimp{v_0}{i_0} \to \vimp{v_1}{i_1} \to \cdots \to \vimp{v_m}{i_m} \to \vimp{v_0}{i_0}$ be a directed cycle in $\Imp{\Instcs}$, where $v_0=v$ and $i_j\in \{0,1\}$, $0\le j \le m$.
		From the construction, $m\ge 1$ holds.
		Without loss of generality, assume that $i_0=0$.
		Recall that each arc $\vimp{v}{i} \to \vimp{w}{\neg{j}}$ means that if $v$ is assigned $i$, then $w$ must be assigned $\neg{j}$.
		Then, $\mapf_0(v_p)=i_p$ holds for every $p\in \ISN{0}{m}$.
		Moreover, by contrapositions of the above implications, $\mapf_1(v_p)=\neg{i_p}$ also holds for every $p\in \ISN{0}{m}$.
		We assume for a contradiction that $\mapf_0$ and $\mapf_1$ are reconfigurable, and consider the first solution $\mapf$ in a reconfiguration sequence such that $\mapf(v_p)=\neg{i_p}$ for some $p\in \ISN{0}{m}$.
		Since there exists an arc $\vimp{v_q}{i_q} \to \vimp{v_p}{i_p}$ in a directed cycle $C$, and hence $\mapf(v_q)$ must be $\neg{i_q}$.
		However, by the definition of $\mapf$, we have $\mapf(v_q)=i_q$, which is a contradiction.
	\end{proof}
	
\subsection{Preprocessing}

	We now explain a preprocessing, which eliminates all ``undesirable'' vertices which prevent the partition from being proper.
	Let $\Inst=(\Instcs,\mapf_\ini,\mapf_\tar)$ be an instance of \CSR{2}, where $\Instcs=(G,\Dom,\Ecns)$.
	Without loss of generality, we can assume that a list $\Vcns(v)$ of every Boolean vertex $v\in \Bs$ is a subset of $\{0,1\}$ by a simple value replacement.
	Then, we define the instance $\Instcs^\resup=(G^\resup,\Dom,\Ecns^\resup)$ of \CSP{2} as the instance obtained by restricting all component of $\Instcs$ on $\Bs$.
	That is, 
	\begin{itemize}
		\item $G^\resup=G[\Bs]$; and
		\item  for each $X^\prime \in E(G^\resup)$, $\Ecns^\resup(X^\prime)=\bigcap_{X\in E^\prime}\subcns{X}$, where $E^\prime=\{X\in E(G) \colon X \cap \Bs=X^\prime \}$ and $\subcns{X}=\{\rest{\mapg}{X^\prime} \colon \mapg \in \Ecns(X)\}$.
	\end{itemize}
	Let $\Vfix$ be the set of vertices $v\in V(G)$ such that $\vimp{v}{0}$ or $\vimp{v}{1}$ are contained in a directed cycle of $\Imp{\Instcs^\resup}$.
	By Lemma~\ref{lem:imp}, in any solution for $\Instcs^\resup$, all vertices $v$ in $\Vfix$ are fixed, that is, cannot be reconfigured at all.
	This property also holds for the original instance $\Instcs$.
	Therefore, if $\rest{\mapf_\ini}{\Vfix}\ne \rest{\mapf_\tar}{\Vfix}$ holds, then we can immediately conclude that $\Inst$ is a no-instance.
	In the other case, we construct in polynomial time the substitution $\sbst{\Instcs}{\maph^\prime}$, where $\maph^\prime=\rest{\mapf_\ini}{\Vfix}=\rest{\mapf_\tar}{\Vfix}$.
	Then, we have the following proposition.
	\begin{proposition}
		$\Inst$ is a yes-instance if and only if $\Inst^\prime=(\sbst{\Instcs}{\maph^\prime},\allowbreak \rest{\mapf_\ini}{V(G^\prime)},\rest{\mapf_\tar}{V(G^\prime)})$ is.
	\end{proposition}
	\begin{proof}
		Because all vertices $v$ in $\Vfix$ are fixed, $\Inst$ is a yes-instance if and only if there exists a reconfiguration sequence $\Seq$ such that every solutions $\mapf$ in $\Seq$ satisfies $\rest{\mapf}{\Vfix}=\maph^\prime$.
		By Lemma~\ref{lem:substitution}, there exists such a reconfiguration sequence if and only if there exists a reconfiguration sequence between $\rest{\mapf_\ini}{V(G^\prime)}$ and $\rest{\mapf_\tar}{V(G^\prime)}$ in $\Sol{\sbst{\Instcs}{\maph^\prime}}$.
	\end{proof}

	Therefore, we can obtain an equivalent instance which satisfies the conditions (1) and (2) of Lemma~\ref{lem:nbcore} by repeating the above transformation until the corresponding implication graph becomes acyclic or empty.
	Since the number of vertices decreases during the process, this can be done in polynomial time.
	Let $\Inst^\ppsup=(\Instcs^\ppsup,\mapf_\ini^\ppsup,\mapf_\tar^\ppsup)$, where $\Instcs^\ppsup=(G^\ppsup,\Dom,\Ecns^\ppsup)$, be an instance obtained by this preprocessing.	
	Then, it is left to prove that $\Inst^\ppsup$ satisfies the condition (3).
	
\subsection{Properity of the pertition}
	
	Let $\Part$ be the partition for $\Instcs^\ppsup$, and let $\Cnd$ be any solution set in $\Part$ such that the restriction of every solution in $\Cnd$ on $\NBs$ is $\maph$.
	By Lemma~\ref{lem:imp}, in order to prove that $\Part$ is proper, it suffices to show that $\Imp{\sbst{\Instcs^\ppsup}{\maph}}$ has no directed cycles.
	
	Assume for a contradiction that $\Imp{\sbst{\Instcs^\ppsup}{\maph}}$ has a cycle $C$.
	Let $\Instcs^{\ppsup \resup}$ be an instance obtained by restricting all component of $\Instcs^\ppsup$ on $\Bs$.
	From the definition of the implication graph, for each vertex $\vimp{v}{i}$ in $C$, there exists a solution $\mapf^\prime$ for $\sbst{\Instcs^\ppsup}{\maph}$ such that $\mapf^\prime(v)=i$.
	By Lemma~\ref{lem:substitution}, a mapping $\mapf$ such that $\rest{\mapf}{\NBs}=\maph$ and $\rest{\mapf}{\Bs}=\mapf^\prime$ is a solution for $\Instcs^\ppsup$.
	Moreover, $\rest{\mapf}{\Bs}=\mapf^\prime$ is a solution for the restricted instance $\Instcs^{\ppsup \resup}$.
	Therefore, $\Imp{\Instcs^{\ppsup \resup}}$ has a vertex $\vimp{v}{i}$.
	For each arc $\vimp{v}{i} \to \vimp{w}{j}$, a mapping $(i, \neg{j})$ is not contained in the constraint $\Ecns^{\ppsup \prime}(vw)$ of $vw$ in $\sbst{\Instcs^\ppsup}{\maph}$.
	Recall that $\Ecns^{\ppsup \prime}(vw)$ is the mapping set 	$\bigcap_{X\in E^\prime}\subcns{X}$, where $E^\prime=\{X\in E(G^\ppsup) \colon X\setminus \NBs=\{v,w\} \}$ and $\subcns{X}=\{\rest{\mapg}{\{v,w\}} \colon \mapg \in \Ecns(X),\allowbreak \mbox{$\maph$ and $\mapg$ are compatible}\}$.
	Since $G^\ppsup$ has a hyperedge of size at most two, $E^\prime$ contains exactly one edge $vw$.
	Therefore, $\Ecns^{\ppsup \prime}(vw)=\Ecns^\ppsup(vw)$ holds, and hence $(i, \neg{j})$ not contained in $\Ecns^\ppsup(vw)$.
	Moreover, the constraint $\Ecns^{\ppsup \resup}(vw)$ does not contain $(i, \neg{j})$, too.
	From the definition, $\Imp{\Instcs^{\ppsup \resup}}$ contains the arc $\vimp{v}{i} \to \vimp{w}{j}$.
	By the above observations, $\Imp{\Instcs^\ppsup}$ has a directed cycle $C$, which contradicts that we have eliminated all directed cycles from the implication graph by the preprocessing.
	
	\medskip
	Thus, we have proved Lemma~\ref{lem:nbcore} and hence Theorem~\ref{the:nb}.
	


\bibliography{biblio}

\end{document}